\documentclass[draftclsnofoot, onecolumn, comsoc, 12pt]{IEEEtran} 

\usepackage{graphicx,epsfig}
\usepackage[noadjust]{cite}
\usepackage{mcite}
\usepackage{amsfonts,helvet}
\usepackage{fancyhdr}
\usepackage{threeparttable}
\usepackage{epsf,epsfig}
\usepackage{amsthm}
\usepackage{amsmath}
\usepackage{siunitx}
\usepackage{amssymb}

\usepackage{dsfont}
\usepackage{subfigure}
\usepackage{color}
\usepackage{algorithm}
\usepackage{algpseudocode}
\usepackage{algcompatible}
\usepackage{enumerate}
\usepackage{gensymb}
\usepackage{cancel}

\newtheorem{corollary}{Corollary}

\newtheorem{lemma}{Lemma}

\newtheorem{proposition}{Proposition}
\newtheorem{remark}{Remark}

\usepackage{eucal}

\newcommand\blfootnote[1]{%
  \begingroup
  \renewcommand\thefootnote{}\footnote{#1}%
  \addtocounter{footnote}{-1}%
  \endgroup
}

\begin{document}

\title{\huge Rate-Splitting Multiple Access for Downlink MIMO: A Generalized Power Iteration Approach}

\author{Jeonghun~Park, Jinseok~Choi, Namyoon~Lee, Wonjae~Shin, and H. Vincent Poor

\thanks{J. Park is with the School of Electronics Engineering, Kyungpook National University, South Korea (e-mail: {\texttt{jeonghun.park@knu.ac.kr}}). J. Choi is with Department of Electrical Engineering, Ulsan National Institute of Science and Technology, South Korea (e-mail: {\texttt{jinseokchoi@unist.ac.kr}}). N. Lee is with Department of Electrical Engineering, POSTECH, South Korea (e-mail: {\texttt{nylee@postech.ac.kr}}). W. Shin is with Department of Electrical and Computer Engineering, Ajou University, South Korea (email: {\texttt{wjshin@ajou.ac.kr}}). H. V. Poor is with Department of Electrical and Computer Engineering, Princeton University, Princeton, NJ, USA. (email: {\texttt{poor@princeton.edu}})
}

\thanks{This work was supported by the National Research Foundation of Korea (NRF) grant funded by the Korea government (MSIT) (No. 2019R1G1A1094703 and No. 2021R1C1C1004438), and the Institute of Information and Communications Technology Planning and Evaluation under Grant
2021-0-00467.}
}

\maketitle \setcounter{page}{1} 

\begin{abstract} 
Rate-splitting multiple access (RSMA) is a general multiple access scheme for downlink multi-antenna systems embracing both classical spatial division multiple access and more recent non-orthogonal multiple access. Finding a linear precoding strategy that maximizes the sum spectral efficiency of RSMA is a challenging yet significant problem. In this paper, we put forth a novel precoder design framework that jointly finds the linear precoders for the common and private messages for RSMA. Our approach is first to approximate the non-smooth minimum function part in the sum spectral efficiency of RSMA using a LogSumExp technique. Then, we reformulate the sum spectral efficiency maximization problem as a form of the log-sum of Rayleigh quotients to convert it into a tractable form. By interpreting the first-order optimality condition {of the reformulated problem} as an eigenvector-dependent nonlinear eigenvalue problem, we reveal that the leading eigenvector of the derived optimality condition is a local optimal solution. To find the leading eigenvector, we propose an algorithm inspired by a power iteration. Simulation results show that the proposed RSMA transmission strategy provides significant improvement in the sum spectral efficiency compared to the state-of-the-art RSMA transmission methods.

\end{abstract}

\begin{IEEEkeywords}
Rate-splitting multiple access (RSMA), multi-user MIMO, imperfect channel state information (CSI), sum spectral efficiency maximization, generalized power iteration.
\end{IEEEkeywords}

\section{Introduction}

\blfootnote{A part of this paper was presented at the Workshop on Rate-Splitting (Multiple Access) for Beyond 5G at the 2021 IEEE Wireless Communications and Networking Conference, Nanjing, China \cite{park:wcnc:21}.} 
Multi-user multiple-input multiple-output (MU-MIMO) downlink transmissions can provide extensive gains in spectral efficiency by serving multiple users with a shared time-frequency resource \cite{caire:tit:03, chris:twc:08, choi:twc:20}. Assuming perfect channel state information at the transmitter (CSIT), a transmitter is able to send information symbols along with multiple linear precoding vectors to different users simultaneously by mitigating inter-user interference. 
In practice, however, the theoretical gains of downlink MU-MIMO transmissions can greatly vanish due to the inaccuracies of the CSIT. 
For example, considering the frequency division duplex (FDD) systems, the downlink channel has to be estimated at the receiver first and sent back to the transmitter via a finite-rate feedback link \cite{jindal:tit:06, park:twc:16}, wherein quantization error on CSIT is inevitable. 
For this reason, in order to attain the de facto MU-MIMO spectral efficiency gains, it is crucial to design a downlink MU-MIMO transmission strategy that achieves high spectral efficiency under imperfect CSIT.

Rate-splitting multiple access (RSMA) is a robust downlink multiple access technique, especially when a transmitter has inaccurate knowledge for downlink CSI. Unlike the conventional spatial division multiple access (SDMA), in RSMA \cite{joudeh:16:tcom, joudeh:16:tsp, dai:16:twc, li:jsac:20}, the transmitter harnesses the rate-splitting strategy that breaks user messages into common and private parts {\color{black}{in order to dynamically manage interference caused by imperfect CSIT.}}
The transmitter constructs a common message by jointly encoding the common parts of the users' split messages. The rate for this common message is carefully controlled so that all the users can decode it. The transmitter also encodes the private parts of the users' messages to generate private information symbols. Then, the transmitter sends the common and private information symbols along with linear precoding vectors in a non-orthogonal manner. Each user decodes and eliminates the common message by performing successive interference cancellation (SIC) while treating the residual interference as noise. It then decodes the desired private message. Thanks to the rate-splitting encoding and SIC decoding, RSMA has been shown to outperform dirty paper coding (DPC) when imperfect CSIT is given \cite{mao:tcom:20}. 

To clearly understand the gains of RSMA over SDMA, it is instructive to consider a simple case of a two-user multi-antenna broadcast channel with imperfect CSIT. From an information-theoretic viewpoint, when applying linear precoding with imperfect CSIT in a two-user multi-antenna broadcast channel, the channel can be interpreted as a virtual two-user interference channel with transmitter cooperation, in which the channel gains of desired and interfering links are determined by the precoding vectors and the channel vectors. In this equivalent interference channel, the quasi-optimal transmission strategy is the Han-Kobayashi scheme \cite{han:tit:81}, i.e., splitting messages into common and private parts and allocating the power according to the relative channel gains between the interfering and the desired link \cite{etkin:tit:08}. Motivated by this, RSMA mimics this near capacity-achieving strategy in downlink MIMO. 


To reap the spectral efficiency gains by using RSMA in downlink MIMO, it is significant to find the optimal linear precoding solution; yet it is challenging to find such a precoding vector. 
Unlike the sum spectral efficiency maximization problem for SDMA relying on private messages only, the problem for RSMA has an additional unique challenge induced by the common message rate, which is the minimum of all the achievable rates for the common message at the users. This minimum function is non-smooth, making the sum-spectral efficiency maximization problem for RSMA challenging to solve. In this work, we put forth a new approach for designing a precoder to maximize the sum spectral efficiency of multi-antenna RSMA with imperfect CSIT.

\subsection{Related Works}


Recently, to cope with imperfect CSIT in downlink MIMO systems, the idea of rate-splitting has been actively re-explored as a multiple access technique, i.e., RSMA \cite{clercks:commmag:16}. 
In \cite{yang:13:tit}, it was shown that RSMA provides sum degrees-of-freedom gains in multi-antenna broadcast channels where erroneous CSIT is given. 
Exploiting the idea of \cite{yang:13:tit}, in \cite{hao:tcom:15}, the achievable spectral efficiency was analyzed while fixing the precoder for the common message as a random precoder and the precoder for the private messages as ZF. 

Besides the theoretical analysis, there exist several prior works that have developed practical linear precoding designs for RSMA multi-antenna systems.
In \cite{joudeh:16:tcom}, a linear precoding design was proposed based on the weighted minimum mean square (WMMSE) approach \cite{chris:twc:08}. Specifically, a non-convex original problem was transformed into a quadratically constrained quadratic program (QCQP) by using the equivalence between the sum spectral efficiency maximization and the sum mean square error (MSE) minimization problem. Subsequently, an interior point method was used to solve the QCQP. Employing the same idea, in \cite{joudeh:16:tsp, yalcin:twc:21}, a max-min fairness problem with RSMA was addressed. 
In \cite{li:jsac:20}, a linear precoding method for general RSMA was proposed by exploiting a concave-convex procedure (CCCP) that successively approximates the original problem into convex forms. 
To evaluate the performance of RSMA in practical settings, e.g., finite constellation and implementable channel coding, \cite{dizdar:pimrc:20} performed a link-level simulation in RSMA downlink MIMO systems. In \cite{yang:icc:20}, considering a single-antenna downlink channel, a power control method was proposed by incorporating the SIC constraint. In \cite{dai:twc:16}, considering downlink massive MIMO, it was shown that hierarchical RSMA that uses multiple-layer partial common messages can be well-harmonized in massive MIMO systems thanks to its spatial covariance separability \cite{adhi:tit:13}. 
Beyond the sum spectral efficiency maximization in downlink, other variants also exist. For example, RSMA for energy efficiency maximization \cite{mao:tcom:19}, RSMA with hardware impairments \cite{papa:tvt:17}, RSMA in joint MIMO radar and communication system \cite{xu:jsac:21},
and RSMA in uplink channels \cite{zeng:tvt:19} have been studied in the context of optimization for RSMA. 
Further, multi-antenna RSMA in interference channels \cite{hao:tcom:17} was also presented.

A key obstacle of the RSMA linear precoding design arises from the common message rate that should be determined as the minimum of all the achievable rates. To resolve this, the conventional methods use convex relaxation. Namely, an original non-convex problem is relaxed into a convex problem first, and then this convexified problem is put into an off-the-shelf optimization toolbox such as CVX to obtain a solution. 
A limitation of such approaches is that the optimization toolbox is hard to implement in practical hardware due to its extremely high complexity \cite{krivochiza:access:19}. For this reason, the existing precoding optimization methods for RSMA are hardly used in practice. Therefore, this paper proposes a new optimization framework for MIMO RSMA that outperforms the existing methods in terms of complexity and also performance. 


\subsection{Contributions}

This paper proposes a new approach for linear precoding optimization in downlink MIMO with RSMA.  
The contributions of this paper are listed as follows. 

\begin{itemize}
    \item Considering an imperfect CSIT model, in which the CSI error statistic is modeled as complex Gaussian with zero-mean and a certain covariance matrix, we derive a lower bound on the instantaneous sum spectral efficiency for RSMA. In contrast to the sum spectral efficiency maximization for SDMA with imperfect CSIT, this lower bound entails the non-smooth minimum function for the common message rate. To convert the non-smooth minimum function into a tractable form, we take the LogSumExp technique, which offers a tight approximation of the minimum function to a smooth function. Then, by representing all optimization variables (precoding vectors) onto a higher dimensional vector, we reformulate the lower bound of the instantaneous sum spectral efficiency for RSMA into a tractable non-convex function in the form of the log-sum of Rayleigh quotients. 

    \item Using the derived lower bound with the smooth function approximation, we establish the first-order optimality condition for the sum spectral efficiency maximization problem. Remarkably, it is shown that the derived condition is cast as an eigenvector-dependent nonlinear eigenvalue problem \cite{cai:siam:18}, where the optimization variable behaves as an eigenvector, and the objective function behaves as an eigenvalue. Accordingly, we reveal that if we find the leading eigenvector that ensures the derived optimality condition, the best local optimal solution is obtained, maximizing the approximate lower bound of the instantaneous sum spectral efficiency for RSMA. 
    
    \item To obtain the leading eigenvector of the derived condition, we put forth a novel algorithm inspired by a power iteration, referred to as generalized power iteration for rate-splitting (GPI-RS). Adopting the conventional power iteration principle, the idea of GPI-RS is to compute the leading eigenvector iteratively. 
    The solution obtained by GPI-RS jointly provides the precoding directions and power allocation for the common and private messages. Notably, we do not rely on CVX in the proposed algorithm; thereby, it is more beneficial to implement in practical hardware. In addition to this, the computational complexity is less compared to the existing WMMSE-based method \cite{joudeh:16:tcom}. 
    Later, we also generalize the proposed GPI-RS for a case in which multiple-layer RSMA is used. In multiple-layer RSMA, not only common message, but also the partial common message that includes messages of a subset of the users are jointly used. We show that the proposed method is suitably extended to this case. 
    
    \item  Simulation results show that the proposed GPI-RS provides spectral efficiency gains over the existing methods, including the conventional convex relaxation-based WMMSE method \cite{joudeh:16:tcom} in various system environments. To be specific, the proposed GPI-RS provides around $20\%$ sum spectral efficiency gains, {\color{black}{while consuming only $6\sim7\%$ of the computation time compared to the conventional method.}} Further, we empirically confirm that the GPI-RS converges well. 
    
    
\end{itemize}

\textit{Notation}:
The superscripts $(\cdot)^{\sf T}$, $(\cdot)^{\sf H}$, and $(\cdot)^{-1}$ denote the transpose, Hermitian, and matrix inversion, respectively. ${\bf{I}}_N$ is the identity matrix of size $N \times N$, Assuming that ${\bf{A}}_1, ..., {\bf{A}}_N \in \mathbb{C}^{K \times K}$, ${\bf{A}} = {\rm blkdiag}\left({\bf{A}}_1, ...,{\bf{A}}_n,..., {\bf{A}}_N \right)$ is a block-diagonal matrix concatenating ${\bf{A}}_1, ..., {\bf{A}}_N$.

\section{System Model} \label{sec:sys_model}

\subsection{Channel Model}
We consider a single-cell downlink MU-MIMO system, where a base station (BS) equipped with $N$ antennas serves $K$ single-antenna users. We denote a user set as $\CMcal{K} = \{1, \cdots ,K\}$. 
The channel vector between the BS and  user $k$ is denoted as ${\bf{h}}_k \in \mathbb{C}^{N}$ for $k \in \CMcal{K}$, where ${\bf{h}}_k$ is generated based on the spatial covariance matrix ${\bf{R}}_k$, i.e., ${\bf{R}}_k = \mathbb{E}\left[ {\bf{h}}_k {\bf{h}}_k^{\sf H} \right]$. 
For constructing the channel covariance matrix, we adopt the one-ring model \cite{adhi:tit:13}. Specifically, we assume that the BS is equipped with uniform circular array with radius $\psi D$ where $\psi$ denotes a signal wavelength and $D = \frac{0.5}{\sqrt{(1-\cos\left( 2\pi/N \right))^2 + \sin^2\left(2\pi/N \right) }}$. Then, the channel correlation coefficient between the $n$-th antenna and $m$-th antenna corresponding to  user $k$ is defined as
\begin{align} \label{eq:spatial_cov}
\left[{\bf{R}}_k \right]_{n,m} = \frac{1}{2\Delta_k} \int_{\theta_k - \Delta_k}^{\theta_k + \Delta_k} e^{-j \frac{2\pi}{\psi} \Psi(x) ({\bf{r}}_n - {\bf{r}}_m)} {\rm d} x,
\end{align}
 where $\theta_k$ is angle-of-arrival (AoA) of  user $k$, $\Delta_k$ is the angular spread of  user $k$, $\Psi(x) = \left[\cos(x), \sin(x) \right]$, and ${\bf{r}}_n$ is the position vector of the $n$-th antenna. By employing the Karhunen-Loeve model as in \cite{adhi:tit:13,dai:twc:16}, the channel vector ${\bf{h}}_k$ is represented as
 \begin{align}
 	{\bf{h}}_k = {\bf{U}}_k \Lambda_k^{\frac{1}{2}} {\bf{g}}_k,
 \end{align}
 where $\Lambda_k \in \mathbb{C}^{r_k \times r_k}$ is a diagonal matrix that contains the non-zero eigenvalues of ${\bf{R}}_k$, ${\bf{U}}_k \in \mathbb{C}^{N \times r_k}$ is a collection of the eigenvectors of ${\bf{R}}_k$ corresponding to the eigenvalues in $\Lambda_k$, and ${\bf{g}}_k \in \mathbb{C}^{r_k}$ is an independent and identically distributed channel vector. We assume that each element of ${\bf{g}}_k$ is drawn from $\mathcal{CN}(0,1)$. We consider a block fading model, where ${\bf{g}}_k$ keeps constant within one transmission block. Over the two consecutive transmission blocks, ${\bf{g}}_k$ changes independently. 
 
 We clarify that the applicability of our method does not depend on particular channel model assumptions. 
 We consider the one-ring model in this paper because it is one of the widely used channel models that suitably captures spatial covariance structures of MIMO channels. The proposed method can be applied in any channel model.

\subsection{CSIT Acquisition Model}
This subsection explains the CSIT estimation and the error model used throughout this paper. 
We assume perfect channel state information at the receiver (CSIR), which can be achieved via downlink pilots planted in the data packet as described in LTE and 5G NR. 
In contrast to CSIR, a BS should estimate CSIT, so that only imperfect knowledge of CSIT is allowed. Generally, two approaches are known to estimate the CSIT, each of which is linear MMSE (LMMSE) and limited feedback, respectively. 
LMMSE yields the optimum CSIT estimation performance, provided that the channel is distributed as Gaussian. Nonetheless, LMMSE only can be used when the channel reciprocity holds. On the contrary to that, limited feedback can be employed in any environment. In this paper, we focus on LMMSE, but note that our method is also useful with limited feedback. 



As mentioned above, LMMSE can be exploited when the BS can use channel reciprocity. Specifically, assuming that the uplink and downlink channels are reciprocal, the BS estimates the CSIT from uplink training sent from the users using the MMSE estimation \cite{yin:jsac:13}. For this reason, LMMSE is adequate to use in time division duplex (TDD) systems where the channel reciprocity holds. Using LMMSE, the estimated CSIT is presented as 
    \begin{align}
        \hat {\bf{h}}_k = {\bf{h}}_k -  {\bf{e}}_k,
    \end{align}
    where $ {\bf{e}}_k$ is the CSIT estimation error vector. Since ${\bf{h}}_k$ is distributed as Gaussian, $\hat {\bf{h}}_{k}$ and ${\bf{e}}_k$ are also Gaussian that is independent to each other. The error covariance is obtained as 
    \begin{align}
        \mathbb{E}[{\bf{e}}_{k} {\bf{e}}_k^{\sf H}] = {\bf{\Phi}}_k = {\bf{R}}_k - {\bf{R}}_k \left({\bf{R}}_k + \frac{\sigma^2}{\tau_{\sf ul} p_{\sf ul}} \right)^{-1} {\bf{R}}_k,
    \end{align}
    where $\tau_{\sf ul}$ and $p_{\sf ul}$ are uplink training length and uplink training transmit power. As the uplink training length and power increases to infinity, the error covariance ${\bf{\Phi}}_k = {\bf{0}}$ and the CSIT error ${\bf{e}}_k$ also vanishes; then the BS has the perfect CSIT.

 \subsection{RSMA Signal Model}
Using RSMA, the message for  user $k$ is split into the common message part $s_{{\sf c},k}$ and the private message $s_{k}$. The common message part $s_{{\sf c}, k}$ from each user is combined to encode the common message $s_{\sf c}$. The common message $s_{\sf c}$ is drawn from a public codebook so that any user associated with the BS can decode it. On the contrary, the private message $s_k$ comes from an individual codebook. Therefore it is only decodable to intended users. 

One common message and $K$ private messages are linearly precoded and then superimposed, so that the transmit signal ${\bf{x}} \in \mathbb{C}^{N}$ is given by
\begin{align} \label{eq:transmit_signal}
{\bf{x}} = {\bf{f}}_{\sf c} s_{\sf c} + \sum_{i = 1}^{K} {\bf{f}}_i s_i,
\end{align}
where ${\bf f}_{\sf c} \in \mathbb{C}^N$ and ${\bf f}_i \in \mathbb{C}^N$ are the precoding vectors for the common and private messages respectively with the transmit power constraint: $\left\|{\bf{f}}_{\sf c} \right\|^2 + \sum_{i=1}^{K} \left\| {\bf{f}}_i \right\|^2 \le 1$. We note that not only the direction of each message, but also the power allocated to each message are controlled by the precoding vectors. For example, if $\| {\bf{f}}_{\sf c} \| = 0$, then no common message is delivered; so our RSMA signal model reduces to typical SDMA. 

The received signal at user $k$ for $k \in \CMcal{K}$ is written as 
\begin{align} \label{eq:rx_signal}
y_{k} = {\bf{h}}_{k}^{\sf H} {\bf{f}}_{\sf c} s_{\sf c} +{\bf{h}}_{k}^{\sf H} {\bf{f}}_{k} s_{k}  +  \sum_{\ell = 1, \ell \neq k}^{K} {\bf{h}}_{k}^{\sf H} {\bf{f}}_{\ell} s_{\ell} + z_k,
\end{align}
where $z_k \sim \mathcal{CN}(0,\sigma^2)$ is additive white Gaussian noise. We also assume that $s_{\sf c}$ and $s_k$ are drawn from an independent Gaussian codebook, i.e., $s_{\sf c}, s_k \sim \mathcal{CN}(0, P)$.

\subsection{Performance Characterization} \label{sec:perf_ch}

To characterize the performance of RSMA, we first explain the decoding process performed by each user. Each user first decodes the common message $s_{\sf c}$ by treating all the other private messages as noise. Once the common message is successfully decoded, using SIC, the users remove the common message from the received signal and decode the private messages with a reduced amount of interference. 

To successfully perform SIC, the common message $s_{\sf c}$ should be decodable to every user without any error. To this end, the code rate of the common message $s_{\sf c}$ is set as the minimum of the ergodic spectral efficiencies among the users. 
Accordingly, under the premise that the BS has imperfect CSIT, i.e., ${\bf \hat h}_{k}={\bf h}_k+{\bf e}_k$ for $k\in \CMcal{K}$, the ergodic spectral efficiency of the common message is obtained as \cite{choi:twc:20, joudeh:16:tcom}
\begin{align} \label{eq:se_common}
R_{\sf c} 
&= \min_{k \in \CMcal{K}} \left\{ \mathbb{E}_{\{ {\bf{h}}_k\}} \left[\left. \log_2 \left(1 + \frac{|{\bf{h}}_{k}^{\sf H} {\bf{f}}_{\sf c}|^2}{\sum_{\ell = 1}^{K} |{\bf{h}}_{k}^{\sf H} {\bf{f}}_\ell|^2 + \sigma^2/P} \right) \right. \right]\right\} \nonumber \\
&= \min_{k \in \CMcal{K}} \left\{ \mathbb{E}_{\{\hat {\bf{h}}_k\}} \left[ \mathbb{E}_{\{ {\bf{e}}_k \}} \left[\left. \log_2 \left(1 + \frac{|{\bf{h}}_{k}^{\sf H} {\bf{f}}_{\sf c}|^2}{\sum_{\ell = 1}^{K} |{\bf{h}}_{k}^{\sf H} {\bf{f}}_\ell|^2 + \sigma^2/P} \right) \right| \hat {\bf{h}}_{k } \right] \right]\right\},
\end{align}
{\color{black}{where in \eqref{eq:se_common}, the inner expectation is taken over the randomness associated with the CSIT error $(\mathbb{E}_{\{ {\bf{e}}_k \}}[\cdot])$ within one particular coherence block and the outer expectation is taken over the randomness associated with the imperfect knowledge of the channel fading process $(\mathbb{E}_{\{ \hat {\bf{h}}_k \}}[\cdot])$.}} 
Assuming that the channel code length spans an infinite number of channel blocks and we set the channel coding rate of the common message $s_{\sf c}$ less than or equal to $R_{\sf c}$, no decoding error for $s_{\sf c}$ occurs. Similar to \eqref{eq:se_common}, the ergodic spectral efficiency of the private message $s_k$ after cancelling the common message $s_{\sf c}$ is obtained as \cite{choi:twc:20, joudeh:16:tcom}
\begin{align} \label{eq:se_private}
R_{k} = \mathbb{E}_{\{\hat {\bf{h}}_k\}} \left[ \mathbb{E}_{\{{\bf{e}}_k \}} \left[\left. \log_2 \left(1 + \frac{|{\bf{h}}_{k}^{\sf H} {\bf{f}}_{k}|^2}{\sum_{\ell = 1, \ell \neq k}^{K} |{\bf{h}}_{k}^{\sf H} {\bf{f}}_\ell|^2 + \sigma^2/P}\right) \right| \hat{\bf{h}}_{k } \right] \right].
\end{align}
Since we assume that the common message is successfully eliminated, we observe that there is no interference from the common message in \eqref{eq:se_private}. 
Under the assumption that the channel code length spans an infinite number of channel blocks and the channel coding rate of the private message $s_k$ is less than or equal to $R_k$, the users can successfully decode $s_k$. 

Our main goal is to optimize the precoders using imperfect knowledge on CSIT per each fading block. For this reason, we focus on one particular fading block without loss of generality, allowing to assume that $\hat {\bf{h}}_k$, $k \in \CMcal{K}$ is given. 
In a certain fading block, we can define the instantaneous spectral efficiency. Specifically, the instantaneous spectral efficiency of the common message achieved at user $k$ is defined as
\begin{align} \label{eq:common_k_ins}
    R_{\sf c}^{\sf ins.}(k) = \mathbb{E}_{\{{\bf{e}}_k\}}\left[\left. \log_2 \left(1 + \frac{|{\bf{h}}_{k}^{\sf H} {\bf{f}}_{\sf c}|^2}{\sum_{\ell = 1}^{K} |{\bf{h}}_{k}^{\sf H} {\bf{f}}_\ell|^2 + \sigma^2/P}\right) \right| \hat{\bf{h}}_{k} \right].
\end{align}
The instantaneous spectral efficiency differs from the ergodic spectral efficiency. On the one hand, the ergodic spectral efficiency is the long-term performance that can be achieved when the channel code length spans very long channel blocks. On the other hand, the instantaneous spectral efficiency is the short-term rate expression when taking into account the channel estimation error effect per channel realization. Considering multiple fading blocks, the instantaneous spectral efficiency and the ergodic spectral efficiency are connected as $R_{\sf c} = \mathop{\min}_{k \in \CMcal{K}} \left\{ \mathbb{E}_{\{\hat {\bf{h}}_k\}} [R_{\sf c}^{\sf ins.}(k)] \right\}$. 

Unfortunately, however, \eqref{eq:common_k_ins} is not tractable. The main challenge is that no closed-form exists for the expectation on CSIT error. To address this, we characterize a lower bound by adopting a similar approach in \cite{choi:twc:20}. We rewrite the received signal \eqref{eq:rx_signal} with the CSIT error term as follows: 
\begin{align} \label{eq:rx_signal_new}
    y_k &= {\bf{h}}_k^{\sf H} {\bf{f}}_{\sf c} s_{\sf c} + \sum_{\ell = 1}^{K} {\bf{h}}_k^{\sf H} {\bf{f}}_{\ell} s_k + z_k \nonumber \\
    &\mathop{=}^{(a)} \hat {\bf{h}}_k^{\sf H} {\bf{f}}_{\sf c} s_{\sf c} + \sum_{\ell  = 1}^{K} \hat {\bf{h}}_k^{\sf H} {\bf{f}}_{\ell} s_k +  \underbrace{ {\bf{e}}_k^{\sf H} {\bf{f}}_{\sf c} s_{\sf c} + \sum_{i = 1}^{K} {\bf{e}}_k^{\sf H} {\bf{f}}_{i} s_i}_{(b)} + z_k,
\end{align}
where (a) follows $ {\bf{h}}_k = \hat {\bf{h}}_k + {\bf{e}}_k$. The term (b) is correlated with the common message $s_{\sf c}$, yet it is not tractable due to the CSIT estimation error ${\bf{e}}_k$. To resolve this, employing a concept of generalized mutual information, we treat (b) as independent Gaussian noise, which is the worst case of mutual information. Then a lower bound on the instantaneous spectral efficiency is made as:
\begin{align}
    \nonumber
     R_{\sf c}^{\sf ins.}(k) & \mathop{\ge}^{(c)}   \mathbb{E}_{\{{\bf{e}}_{k}\}}  \left[  \log_2 \left( 1 + \frac{|\hat {\bf{h}}_{k}^{\sf H} {\bf{f}}_{\sf c}|^2}{\sum_{\ell = 1}^{K} |\hat {\bf{h}}_{k}^{\sf H} {\bf{f}}_\ell|^2 + |{\bf{e}}_k^{\sf H} {\bf{f}}_{\sf c}|^2 + \sum_{\ell = 1}^{K} |{\bf{e}}_k^{\sf H} {\bf{f}}_{\ell}|^2  +  \frac{\sigma^2}{P}} \right)   \right] \nonumber \\
     & \mathop{\ge}^{(d)}     \log_2 \left( 1 + \frac{|\hat {\bf{h}}_{k}^{\sf H} {\bf{f}}_{\sf c}|^2}{\sum_{\ell = 1}^{K} |\hat {\bf{h}}_{k}^{\sf H} {\bf{f}}_\ell|^2 + {\bf{f}}_{\sf c}^{\sf H } \mathbb{E}  \left[ {\bf{e}}_k {\bf{e}}_k^{\sf H}   \right] {\bf{f}}_{\sf c}+ \sum_{\ell = 1}^{K} {\bf{f}}_{\ell}^{\sf H} \mathbb{E}  \left[{\bf{e}}_k {\bf{e}}_k^{\sf H} \right] {\bf{f}}_\ell  +  \frac{\sigma^2}{P}} \right)    \nonumber \\
    & \mathop{=}^{(e)}  \log_2 \left( 1 + \frac{|\hat {\bf{h}}_{k}^{\sf H} {\bf{f}}_{\sf c}|^2}{\sum_{\ell = 1}^{K} |\hat {\bf{h}}_{k}^{\sf H} {\bf{f}}_\ell|^2 + {\bf{f}}_{\sf c}^{\sf H} {\bf{\Phi}}_k {\bf{f}}_{\sf c} + \sum_{\ell = 1}^{K} {\bf{f}}_{\ell}^{\sf H} {\bf{\Phi}}_k {\bf{f}}_{\ell}  +  \frac{\sigma^2}{P}} \right) = \bar R_{\sf c}^{\sf ins.} (k), \label{eq:common}
\end{align}
where (c) comes from treating (b) in \eqref{eq:rx_signal_new} as independent Gaussian noise, (d) follows Jensen's inequality, and (e) comes from the CSIT error covariance $\mathbb{E}[{\bf{e}}_k {\bf{e}}_k^{\sf H}] = {\bf{\Phi}}_k$.
A lower bound on the instantaneous spectral efficiency, denoted as $\bar R_{\sf c}^{\sf ins.} (k)$, is derived as a closed-form, so that this can be handled easily. 
Finally, we take another lower bound on the ergodic spectral efficiency. The ergodic spectral efficiency of the common message is represented with $\bar R_{\sf c}^{\sf ins.} (k)$ as 
\begin{align} \label{eq:ergodic_se_inequal_common}
R_{\sf c} = \min_{k \in \CMcal{K}} \left\{ \mathbb{E}_{\{\hat {\bf{h}}_k\}}\left[  R_{\sf c}^{\sf ins.}(k) \right] \right\} &\ge \min_{k \in \CMcal{K}} \left\{ \mathbb{E}_{\{ \hat {\bf{h}}_k \}}\left[ \bar R_{\sf c}^{\sf ins.}(k) \right]  \right\} \nonumber \\
& \mathop{\ge}^{(f)} \mathbb{E}_{\{ \hat {\bf{h}}_{k \in \CMcal{K}} \}}\left[\min_{k \in \CMcal{K}} \left\{\bar R_{\sf c}^{\sf ins.}(k) \right\} \right]
\end{align}
where (f) follows the fact that putting the minimum operator into the expectation does not increase the value. We take \eqref{eq:ergodic_se_inequal_common} as our main objective function for the common message rate. 

Next, we characterize a lower bound on the instantaneous spectral efficiency of the private message. We first define the instantaneous spectral efficiency of the private message $s_k$ in a certain fading block as $R^{\sf ins.}_k$. Using a similar technique to the common message case, we derive a lower bound on $R^{\sf ins.}_{k}$ such as 
\begin{align} \label{eq:se_ins_private}
    \nonumber
    R_k^{\sf ins.} & \mathop{\ge}^{} \mathbb{E}_{\{{\bf{e}}_{k} \}} \Bigg[\log_2 \left(1 + \frac{|\hat {\bf{h}}_{k}^{\sf H} {\bf{f}}_{k}|^2}{\sum_{\ell = 1, \ell \neq k}^{K} |\hat {\bf{h}}_{k}^{\sf H} {\bf{f}}_\ell|^2 + \sum_{\ell = 1}^{K} |{\bf{e}}_k^{\sf H} {\bf{f}}_{\ell}|^2 +
    \frac{\sigma^2}{P}}\right) \Bigg] \nonumber \\
    & \mathop{\ge }^{} \log_2 \left(1 + \frac{|\hat {\bf{h}}_{k}^{\sf H} {\bf{f}}_{k}|^2}{\sum_{\ell = 1, \ell \neq k}^{K} |\hat {\bf{h}}_{k}^{\sf H} {\bf{f}}_\ell|^2 + \sum_{\ell = 1}^{K}{\bf{f}}_{\ell}^{\sf H} {\bf{\Phi}}_k {\bf{f}}_{\ell} + \frac{\sigma^2}{P}}\right) = \bar R_k^{\sf ins.},
\end{align}
Considering multiple fading blocks, the obtained lower bound on the instantaneous spectral efficiency $\bar R_k^{\sf ins.}$ is connected to the ergodic spectral efficiency as follows:
\begin{align} \label{eq:private_inequal}
R_k = \mathbb{E}_{\{\hat {\bf{h}}_k\}} [R_k^{\sf ins.}] \ge \mathbb{E}_{\{ \hat {\bf{h}}_k \}}[\bar R_k^{\sf ins.}].    
\end{align}
Combining \eqref{eq:ergodic_se_inequal_common} and \eqref{eq:private_inequal}, we finally complete the following lower bound on the ergodic sum spectral efficiency $R_{\Sigma}$:
\begin{align} \label{eq:sum_se_lb}
    R_{\Sigma} & \ge \bar R_{\Sigma} = 
    \mathbb{E}_{\{ \hat {\bf{h}}_{k \in \CMcal{K}}\}}\left[\min_{k \in \CMcal{K }} \{\bar R_{\sf c}^{\sf ins.}(k)\} + \sum_{k = 1}^{K} \bar R_{k}^{\sf ins.}  \right]. 
\end{align}
Now we are ready to formulate our main problem. 

\begin{remark} [Comparison to \cite{li:jsac:20}] \normalfont

Similar to our lower bound, \cite{li:jsac:20} also proposed a lower bound on the instantaneous spectral efficiencies by incorporating the CSIT estimation error. To be specific, \cite{li:jsac:20} assumed a particular scenario of CSIT estimation that ${\bf{h}}_k = \hat {\bf{h}}_k + \delta {\bf{e}}_k$, where $\mathbb{E}[{\bf{e}}_k^{\sf H} {\bf{e}}_k] = {\bf{I}}$ and $\mathbb{E}[{\bf{e}}_k] = 0$. We note that this is a special case of our CSIT estimation model. 
Under this premise, the instantaneous spectral efficiency of the common message achieved at user $k$, $R_{\sf c}^{\sf ins.}(k)$, is expressed as 
\begin{align} \label{eq:common_k_ins_cccp}
    R_{\sf c}^{\sf ins.}(k) &= \mathbb{E}_{\{{\bf{e}}_k\}}\left[\left. \log_2 \left(1 + \frac{{\bf{h}}_k^{\sf H} {\bf{Q}}_{\sf c} {\bf{h}}_k }{\sum_{\ell = 1}^{K} {\bf{h}}_k^{\sf H} {\bf{Q}}_{\ell} {\bf{h}}_k + \sigma^2/P}\right) \right| \hat{\bf{h}}_{k} \right] \nonumber \\
    &= \mathbb{E}_{\{{\bf{e}}_k\}} \left[ \log_2\left( {\bf{h}}_k^{\sf H} {\bf{Q}}_{\sf c} {\bf{h}}_k + \sum_{\ell = 1}^{K} {\bf{h}}_k^{\sf H} {\bf{Q}}_\ell {\bf{h}}_k + \sigma^2/P \right) \right] - \mathbb{E}_{\{{\bf{e}}_k\}} \left[ \log_2 \left( \sum_{\ell = 1}^{K} {\bf{h}}_k^{\sf H} {\bf{Q}}_\ell {\bf{h}}_k + \sigma^2/P \right) \right],
\end{align} 
where ${\bf{Q}}_{\sf c} = {\bf{f}}_{\sf c} {\bf{f}}_{\sf c}^{\sf H}$ and ${\bf{Q}}_k = {\bf{f}}_k {\bf{f}}_k^{\sf H}$. 
By using Jensen's inequality and the fact that the CSIT estimation error has zero-mean, we obtain an upper bound on the second term in \eqref{eq:common_k_ins_cccp} as 
\begin{align} \label{eq:cccp_lb_second}
    \mathbb{E}_{\{{\bf{e}}_k\}} \left[ \log_2 \left( \sum_{\ell = 1}^{K} {\bf{h}}_k^{\sf H} {\bf{Q}}_\ell {\bf{h}}_k + \sigma^2/P \right) \right] \le \log_2 \left(\sum_{\ell = 1}^{K} ( \hat {\bf{h}}_k^{\sf H} {\bf{Q}}_\ell \hat {\bf{h}}_k + \delta^2{\sf{tr}} ({\bf{Q}}_\ell)) + \sigma^2/P \right)
\end{align}
Additionally, \cite{li:jsac:20} proved that the first term in \eqref{eq:common_k_ins_cccp} is non-decreasing with $\delta$, so we get a lower bound by putting $\delta = 0$, yielding
\begin{align} \label{eq:cccp_lb_first}
    \mathbb{E}_{\{{\bf{e}}_k\}} \left[ \log_2\left( {\bf{h}}_k^{\sf H} {\bf{Q}}_{\sf c} {\bf{h}}_k + \sum_{\ell = 1}^{K} {\bf{h}}_k^{\sf H} {\bf{Q}}_\ell {\bf{h}}_k + \sigma^2/P \right) \right] \ge \log_2\left( \hat {\bf{h}}_k^{\sf H} {\bf{Q}}_{\sf c} \hat {\bf{h}}_k + \sum_{\ell = 1}^{K} \hat {\bf{h}}_k^{\sf H} {\bf{Q}}_\ell \hat {\bf{h}}_k + \sigma^2/P \right).
\end{align}
Combining \eqref{eq:cccp_lb_second} and \eqref{eq:cccp_lb_first}, \cite{li:jsac:20} claimed that we can make a lower bound as 
\begin{align} \label{eq:cccp_csit_lb}
    R_{\sf c}^{\sf ins.}(k) & \ge \log_2 \left(1 + \frac{ \hat {\bf{h}}_k^{\sf H} {\bf{Q}}_{\sf c} \hat {\bf{h}}_k }{\sum_{\ell = 1}^{K} (\hat {\bf{h}}_k^{\sf H} {\bf{Q}}_{\ell} \hat {\bf{h}}_k + \delta^2 {\sf{tr}}({\bf{Q}}_\ell) ) + \sigma^2/P}\right)= \tilde R_{\sf c}^{\sf ins.}(k). 
\end{align}
The lower bound \eqref{eq:cccp_csit_lb} is closely related to our lower bound. Rewriting \eqref{eq:cccp_csit_lb}, we have 
\begin{align} \label{eq:cccp_csit_lb_otherform}
    \tilde R_{\sf c}^{\sf ins.}(k) = \log_2 \left( 1 + \frac{|\hat {\bf{h}}_{k}^{\sf H} {\bf{f}}_{\sf c}|^2}{\sum_{\ell = 1}^{K} |\hat {\bf{h}}_{k}^{\sf H} {\bf{f}}_\ell|^2 + \sum_{\ell = 1}^{K} {\bf{f}}_{\ell}^{\sf H} \cdot \delta^2 {\bf{I}}  \cdot {\bf{f}}_{\ell}  +  \sigma^2/P }\right).
\end{align}
Comparing $\tilde R_{\sf c}^{\sf ins.}(k)$ in \eqref{eq:cccp_csit_lb_otherform} to our lower bound $\bar R_{\sf c}^{\sf ins.}(k)$ in \eqref{eq:common} under the assumption that ${\bf{\Phi}}_k = \delta^2 {\bf{I}}$, \eqref{eq:cccp_csit_lb_otherform} seems to be tighter since the denominator of the SINR in \eqref{eq:cccp_csit_lb_otherform} does not include the term ${\bf{f}}_{\sf c}^{\sf H} \cdot \delta^2 {\bf{I}} \cdot {\bf{f}}_{\sf c}$. Nonetheless, the lower bound \eqref{eq:cccp_csit_lb_otherform} is limited in its applicability. This is because, the lower bound technique to derive \eqref{eq:cccp_csit_lb_otherform} cannot be applied when ${\bf{\Phi}}_k \neq \delta^2{\bf{I}}$, i.e., the CSIT error is spatially correlated. Since it is usual that MIMO channels have particular spatial correlation structures, our lower bound is more proper to use in general cases.

\end{remark}

\subsection{Problem Formulation}

We aim to maximize $\bar R_{\Sigma}$ in \eqref{eq:sum_se_lb}. 
In our setup, maximizing $\bar R_{\Sigma}$ is equivalent to maximizing 
$\min_{k \in \CMcal{K }} \{\bar R_{\sf c}^{\sf ins.}(k)\} + \sum_{k = 1}^{K} \bar R_{k}^{\sf ins.}  $ per each fading block, wherein the BS is able to calculate $\bar R_k^{\sf ins.}$ and $\min_{k \in \CMcal{K }} \{\bar R_{\sf c}^{\sf ins.}(k)\}$ 
in a closed-form by using the estimated CSIT. 
Accordingly, we formulate an optimization problem as follows:
\begin{align} \label{eq:lb_problem}
\mathop{{\text{maximize}}}_{{\bf{f}}_{\sf c}, {\bf{f}}_1, \cdots,{\bf{f}}_K}& \;\; \min_{k \in \CMcal{K }} \{\bar R_{\sf c}^{\sf ins.}(k)\} + \sum_{k = 1}^{K} \bar R_{k}^{\sf ins.} \\
{\text{subject to}} & \;\;  \left\| {\bf{f}}_{\sf c} \right\|^2 + \sum_{k = 1}^{K} \left\| {\bf{f}}_k \right\|^2  \le 1 .\label{eq:lb_problem_constraint}
\end{align}
We tackle \eqref{eq:lb_problem} as our main problem. 
Finding the global solution of \eqref{eq:lb_problem} is infeasible due to its non-convexity and non-smoothness. 

\section{Existing Approach: WMMSE} \label{sec:preliminary}
In this section, we briefly introduce the existing WMMSE approach \cite{joudeh:16:tcom} for the sum spectral efficiency maximization in an RSMA scenario. 
We focus on two points: {\it{(i)}} how to solve a non-convex optimization problem and {\it{(ii)}} how to incorporate the CSIT estimation error. Then we explain the main distinguishable points of the proposed method. 

We first explain how to relax a non-convex problem. 
To solve a non-convex sum spectral efficiency maximization problem, \cite{joudeh:16:tcom} adopted a well-known WMMSE relaxation technique \cite{chris:twc:08}. Specifically, we denote that $\hat s_{\sf c}(k)$, $\hat s_k$ are estimates for $s_{\sf c}(k)$ and $s_k$, where $s_{\sf c}(k)$ is the common message received at user $k$. With scalar equalizers $g_{\sf c}(k)$ and $g_k$, we define the MSEs of the common message received at user $k$ ($\epsilon_{\sf c}(k)$) and the private message ($\epsilon_k$) as
\begin{align}
    & \epsilon_{\sf c}(k) = \mathbb{E} [| \hat s_{\sf c}(k) - s_{\sf c}|^2] = \mathbb{E} [|g_{\sf c}(k)y_k - s_{\sf c}|^2] \nonumber \\
     & \;\;\;\;\;\;\;\; = |g_{\sf c}(k)|^2 T_{\sf c}(k) - 2{\sf{Re}} \left\{ 
    g_{\sf c}(k) {\bf{h}}_k^{\sf H} {\bf{f}}_{\sf c}\right\} + 1, \\
    &\epsilon_k = |g_k|^2 T_k - 2 {\sf{Re}} \left\{ g_k {\bf{h}}_k^{\sf H} {\bf{f}}_k \right\} + 1,
\end{align}
where $T_{\sf c}(k) = |{\bf{h}}_{k}^{\sf H} {\bf{f}}_{\sf c}|^2 + \sum_{\ell = 1}^{K} |{\bf{h}}_k^{\sf H} {\bf{f}}_\ell|^2 + \sigma^2$ and $T_k = |{\bf{h}}_k^{\sf H} {\bf{f}}_k|^2 + \sum_{\ell = 1, \ell \neq k}^{K} |{\bf{h}}_k^{\sf H} {\bf{f}}_\ell|^2 + \sigma^2$.
Note that here we exchange the power assumption between the message and the precoding vectors, i.e., $\mathbb{E}[|s_{\sf c}|^2] = \mathbb{E}[|s_{k}|^2] = 1$ and $\| {\bf{f}}_{\sf c}\|^2 + \sum_{k = 1}^{K} \| {\bf{f}}_k\|^2 \le P$ for ease of description. This does not change the SINR. 
The minimum MSEs are achieved when $g_{\sf c}(k) = {\bf{f}}_{\sf c}^{\sf H} {\bf{h}}_k T_{\sf c}(k)^{-1}$ and $g_k = {\bf{f}}_k^{\sf H} {\bf{h}}_k T_k^{-1}$, providing the following MMSE: $\epsilon_{\sf c}^{\sf MMSE}(k) = T_{\sf c}(k)^{-1} (T_{\sf c}(k) - |{\bf{h}}_k^{\sf H} {\bf{f}}_{\sf c}|^2)$ and $\epsilon_k^{\sf MMSE} = T_k^{-1} (T_k - |{\bf{h}}_k^{\sf H} {\bf{f}}_k|^2)$.

Then the augmented WMSEs are give by
\begin{align}
    & \xi_{\sf c}(k) = u_{\sf c}(k) \epsilon_{\sf c}(k) - \log_2 (u_{\sf c}(k)) \nonumber \\
    & \;\;\;\;\;\;\;\; = {\bf{f}}_{\sf c}^{\sf H} \left(u_{\sf c}(k) |g_{\sf c}(k)|^2 {\bf{h}}_k {\bf{h}}_k^{\sf H} \right) {\bf{f}}_{\sf c} + \sum_{\ell = 1}^{K} {\bf{f}}_\ell^{\sf H} \left(u_{\sf c}(k) |g_{\sf c}(k)|^2 {\bf{h}}_k {\bf{h}}_k^{\sf H} \right) {\bf{f}}_\ell - 2{\sf{Re}}\left\{ u_{\sf c}(k)  g_{\sf c}(k) {\bf{h}}_k^{\sf H} {\bf{f}}_{\sf c}\right\} \nonumber \\
    & \;\;\;\;\;\;\;\;\;\;\; + \sigma^2 u_{\sf c}(k) |g_{\sf c}(k)|^2 + u_{\sf c}(k) - \log_2 (u_{\sf c}(k)),\\
& \xi_k = u_k \epsilon_k - \log_2 (u_k) \nonumber \\
    & \;\;\;\; = {\bf{f}}_k^{\sf H} \left( u_k |g_k|^2 {\bf{h}}_k {\bf{h}}_k^{\sf H} \right) {\bf{f}}_k + \sum_{\ell = 1, \ell \neq k}^{K} {\bf{f}}_\ell^{\sf H} \left(u_k |g_k|^2 {\bf{h}}_k {\bf{h}}_k^{\sf H} \right) {\bf{f}}_\ell - 2{\sf{Re}}\left\{ u_k g_k {\bf{h}}_k^{\sf H} {\bf{f}}_k \right\}  \nonumber \\
    & \;\;\;\;\;\;\; + \sigma^2 u_k |g_k|^2 + u_k - \log_2(u_k).
\end{align}
We note that the optimal weights to achieve the minimum of $ \xi_{\sf c}(k)$ and $ \xi_k$ are obtained as $u_{\sf c}(k) = 1/\epsilon_{\sf c}^{\sf MMSE}(k)$ and $u_k = 1/\epsilon_k^{\sf MMSE}$. 
Upon this weight update, by the rate-WMMSE equivalence \cite{chris:twc:08}, the sum spectral efficiency is maximized by solving the following WMSE minimization problem:
\begin{align} \label{eq:wmmse}
\mathop{{\text{minimize}}}_{{\bf{f}}_{\sf c}, {\bf{f}}_1, \cdots,{\bf{f}}_K, \xi_{\sf c}}& \;\; 
\xi_{\sf c} + \sum_{k = 1}^{K} \xi_k
\\
{\text{subject to}} & \;\;  \xi_{\sf c}(k) \le \xi_{\sf c}, \; \forall k \in \CMcal{K}, \label{eq:wmmse_constraint}\\
& \| {\bf{f}}_{\sf c}\|^2 + \sum_{k = 1}^{K} \| {\bf{f}}_k \|^2 \le P.
\end{align}
The problem \eqref{eq:wmmse} is QCQP, which can be solved by using CVX. We compute the optimal weight, equalizer, and the precoding vector in an alternating fashion. We repeat this process until a certain termination criterion is met. 

The presented WMMSE approach assumes the perfect CSIT. To take the CSIT estimation error into account, \cite{joudeh:16:tcom} adopted the sample average approximation (SAA) technique. In the SAA technique, we produce $M$ number of samples of the augmented WMSEs by randomly generating channel vector ${\bf{h}} = \hat {\bf{h}}_k + {\bf{e}}_k$ ($\hat {\bf{h}}_k$ is given, ${\bf{e}}_k$ is randomly generated) then calculate an empirical average of these samples as follows: 
\begin{align}
    & \bar \xi_{\sf c}(k) \leftarrow \frac{1}{M} \sum_{m = 1}^{M} \xi_{\sf c}^{(m)}(k), \\
    & \bar \xi_{k} \leftarrow \frac{1}{M} \sum_{m = 1}^{M} \xi_{\sf c}^{(m)}(k).
\end{align}
We replace $\xi_{\sf c}(k)$ and $\xi_k$ by the empirical average augmented WMSE $\bar \xi_{\sf c}(k)$ and $\bar \xi_k$ in \eqref{eq:wmmse}. 

Now we clarify the distinguishable points of our method compared to the WMMSE approach. In the WMMSE approach, CVX is required to solve the WMSE minimization problem \eqref{eq:wmmse}. We need additional efforts to implement CVX in FPGA hardware since CVX is not designed to run in real-time hardware \cite{krivochiza:access:19}. In addition to this, we observe that the common message rate is controlled by the WMSE constraints \eqref{eq:wmmse_constraint}, where the number of the constraints is equal to the number of users $K$. For this reason, the associated computational complexity of the WMMSE approach scales with $K^{3.5}$ {\color{black}{\cite{patil:tcom:18, joudeh:16:tsp}}}; resulting in the huge computational complexity is caused when there are a large number of users. 
Compared to this, as we will show later, our method does not require CVX to obtain a solution. Further, since we control the common message rate by cleverly approximating the minimum function, the computational complexity scales with $K$. For this reason, our method is much more beneficial when there are many users.


\section{Precoder Optimization with Generalized Power Iteration} \label{sec:main}

In this section, we explain the key ideas to solve the optimization problem \eqref{eq:lb_problem}. We first approximate the non-smooth minimum function as a smooth function using the LogSumExp technique. Subsequently, we represent the optimization variable onto a higher dimensional vector to reformulate the problem \eqref{eq:lb_problem} into a tractable non-convex optimization problem expressed as a function of Rayleigh quotients. 
By deriving the first-order KKT condition for the reformulated problem, we show that the first-order optimality condition is cast as an eigenvector-dependent nonlinear eigenvalue problem (NEPv) \cite{cai:siam:18}, and finding the leading eigenvector is equivalent to finding the best local optimal point of the reformulated problem. Consequently, to find the leading eigenvector, we propose a computationally efficient generalized power iteration algorithm. 

\subsection{Reformulation to a Tractable Form}
At first, we approximate the non-smooth minimum function by using the LogSumExp technique. With the LogSumExp, the minimum function is approximated as \cite{shen2010dual}
\begin{align}
    \label{eq:logsumexp}
    \min_{i = 1,...,N}\{x_i\}  \approx -\alpha \log\left(\frac{1}{N}  \sum_{i = 1}^{N} \exp\left( \frac{x_i}{-\alpha}  \right)\right),
\end{align}
where the approximation becomes tight as $\alpha \rightarrow +0$. Leveraging \eqref{eq:logsumexp}, we approximate  
\begin{align} \label{eq:approximation_R}
\min_{k \in \CMcal{K }} \{\bar R_{\sf c}^{\sf ins.}(k)\} \approx  -\alpha \log \left(\frac{1}{K} \sum_{k = 1}^{K} \exp\left( \frac{\bar R_{\sf c}^{\sf ins.}(k)}{-\alpha}\right) \right).
\end{align}
To help understand the LogSumExp approximation technique, we draw an illustration in Fig.~\ref{fig:logsumexp_verify}. 
In Fig.~\ref{fig:logsumexp_verify}, assuming the $N = 1$, $K = 2$, a landscape of the minimum spectral efficiency between two users is depicted. In addition to that, the approximate minimum spectral efficiency using the LogSumExp technique is also presented. As shown in the figure, the true maximum value of the non-smooth minimum function is tightly approximated by the LogSumExp technique. 
\begin{figure}[!t]
	\centerline{\resizebox{0.45\columnwidth}{!}{\includegraphics{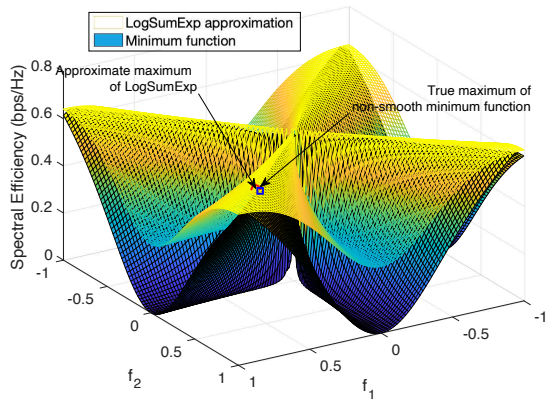}}}     
	\caption{An illustration of the comparison between the approximate maximum using the LogSumExp and the true maximum of the non-smooth minimum function.}
 	\label{fig:logsumexp_verify}
\end{figure}

Now we rewrite the precoding vectors ${\bf{f}}_{\sf c}, {\bf{f}}_1, \cdots, {\bf{f}}_K$ in a higher dimensional vector $\bar {\bf{f}}$ by stacking each vector as 
$\bar {\bf{f}} = [{\bf{f}}_{\sf c}^{\sf T}, {\bf{f}}_1^{\sf T}, \cdots, {\bf{f}}_K^{\sf T}]^{\sf T} \in \mathbb{C}^{N(K+1) \times 1}$.  
With this, we express the instantaneous spectral efficiencies regarding the common message $s_{\sf c}$ as
\begin{align} \label{eq:rewrite_commonrate}
\bar R_{\sf c}^{\sf ins.}(k) = \log_2 \left( \frac{\bar {\bf{f}}^{\sf H} {\bf{A}}_{\sf c} (k) \bar {\bf{f}}}{\bar {\bf{f}}^{\sf H} {\bf{B}}_{\sf c}(k)  \bar {\bf{f}} } \right),
\end{align}
where 
\begin{align}
&{\bf{A}}_{\sf c}(k) =  {\rm blkdiag} \left((\hat {\bf{h}}_k \hat {\bf{h}}_k^{\sf H} + {\bf{\Phi}}_k) , \cdots , (\hat {\bf{h}}_k \hat {\bf{h}}_k^{\sf H} + {\bf{\Phi}}_k) \right) + {\bf{I}}_{N(K+1)} \frac{ \sigma^2}{P}, \\
& {\bf{B}}_{\sf c}(k) = {\bf{A}}_{\sf c}(k) - {\rm blkdiag} \left( \hat {\bf{h}}_k \hat {\bf{h}}_k^{\sf H}, {\bf{0}}, \cdots, {\bf{0}} \right) .
\end{align}
%
Note that we implicitly assume $\|\bar {\bf{f}}\|^2 = 1$ to have \eqref{eq:rewrite_commonrate}. This assumption does not hurt the optimality since the spectral efficiency is monotonically increasing with the transmit power. It is also worthwhile to mention that \eqref{eq:rewrite_commonrate} is presented as a function of Rayleigh quotient terms. Similar to this, we also write the spectral efficiency of the private message $s_k$ as 
\begin{align}
\bar R_{k}^{\sf ins.} = \log_2\left( \frac{\bar {\bf{f}}^{\sf H} {\bf{A}}_k \bar {\bf{f}}}{\bar {\bf{f}}^{\sf H} {\bf{B}}_k \bar {\bf{f}}}   \right),
\end{align}
where 
\begin{align}
&{\bf{A}}_k =  {\rm blkdiag} \left({\bf{0}}, (\hat {\bf{h}}_k \hat {\bf{h}}_k^{\sf H} + {\bf{\Phi}}_k) , \cdots, (\hat {\bf{h}}_k \hat {\bf{h}}_k^{\sf H} + {\bf{\Phi}}_k) \right)+ {\bf{I}}_{N(K+1)} \frac{ \sigma^2}{P}, \\
& {\bf{B}}_{k} =  {\bf{A}}_{k} - {\rm blkdiag} \left({\bf{0}}, \cdots, {\bf{0}},  \underbrace{\hat {\bf{h}}_k \hat {\bf{h}}_k^{\sf H}}_{{\text{the} \;(k+1){\text{th block}}}}, {\bf{0}}, \cdots, {\bf{0}} \right).
\end{align}
With this Rayleigh quotients representation, the problem \eqref{eq:lb_problem} is transformed to 
\begin{align} \label{eq:problem_new}
\mathop{{\text{maximize}}}_{\bar {\bf{f}}} \;\;  
 &\log \left(\frac{1}{K}  \sum_{k = 1}^{K}  \exp \left(\log_2 \left( \frac{\bar {\bf{f}}^{\sf H} {\bf{A}}_{\sf c} (k) \bar {\bf{f}}}{\bar {\bf{f}}^{\sf H} {\bf{B}}_{\sf c}(k)  \bar {\bf{f}} }  \right)^{-\frac{1}{\alpha}} \right)  \right)^{-\alpha}  + \sum_{k = 1}^{K}  \log_2 \left( \frac{\bar {\bf{f}}^{\sf H} {\bf{A}}_k \bar {\bf{f}}}{\bar {\bf{f}}^{\sf H} {\bf{B}}_k \bar {\bf{f}}}   \right) \\
\label{eq:constraint_new}
{\text{subject to}} & \;\; \left\| \bar {\bf{f}} \right\|^2 = 1.
\end{align}
We note that in \eqref{eq:problem_new}, the obtained precoding vector $\bar {\bf{f}}$ can always be normalized by dividing the numerator and the denominator of each Rayleigh quotient with $\| \bar {\bf{f}} \|$, while not affecting the objective function. Thanks to this feature, the constraint $\| {\bf{f}} \|^2 = 1$ can vanish from \eqref{eq:constraint_new}. Now we are ready to tackle the problem \eqref{eq:problem_new}. 

\subsection{First-Order Optimality Condition}
To approach the solution of the transformed problem \eqref{eq:problem_new}, we derive a first-order optimality condition of \eqref{eq:problem_new}. The following lemma shows the main result in this subsection. 
 
  \begin{lemma} \label{lem:main}
 The first-order optimality condition of the optimization problem \eqref{eq:problem_new} is satisfied if the following holds:
\begin{align} \label{eq:lem_kkt}
{\bf{B}}_{\sf KKT}^{-1} (\bar {\bf{f}}){\bf{A}}_{\sf KKT} (\bar {\bf{f}}) \bar {\bf{f}} = \lambda (\bar {\bf{f}})  \bar {\bf{f}},
\end{align} 
where 
\begin{align} \label{eq:lem_A_kkt}
    &{\bf{A}}_{\sf KKT}(\bar {\bf{f}}) =  \lambda_{\sf num} (\bar {\bf{f}}) \times \sum_{k = 1}^{K}  \left[ \frac{\exp\left( \frac{1}{-\alpha}  \log_2\left(\frac{\bar {\bf{f}}^{\sf H} {\bf{A}}_{\sf c}(k) \bar {\bf{f}}}{\bar {\bf{f}}^{\sf H} {\bf{B}}_{\sf c}(k) \bar {\bf{f}} } \right) \right)}{\sum_{\ell = 1}^{K} \exp\left(\frac{1}{-\alpha} \log_2\left(\frac{\bar {\bf{f}}^{\sf H} {\bf{A}}_{\sf c}(\ell) \bar {\bf{f}}}{\bar {\bf{f}}^{\sf H} {\bf{B}}_{\sf c}(\ell) \bar {\bf{f}} } \right)\right)} \frac{{\bf{A}}_{\sf c}(k)}{\bar {\bf{f}}^{\sf H} {\bf{A}}_{\sf c}(k) \bar {\bf{f}}} + \frac{{\bf{A}}_k}{\bar {\bf{f}}^{\sf H} {\bf{A}}_k \bar {\bf{f}}} \right] , \\
    &{\bf{B}}_{\sf KKT}(\bar {\bf{f}}) = \lambda_{\sf den} (\bar {\bf{f}}) \times  \sum_{k = 1}^{K}  \left[ \frac{\exp\left( \frac{1}{-\alpha}  \log_2\left(\frac{\bar {\bf{f}}^{\sf H} {\bf{A}}_{\sf c}(k) \bar {\bf{f}}}{\bar {\bf{f}}^{\sf H} {\bf{B}}_{\sf c}(k) \bar {\bf{f}} } \right) \right)}{\sum_{\ell = 1}^{K} \exp\left(\frac{1}{-\alpha} \log_2\left(\frac{\bar {\bf{f}}^{\sf H} {\bf{A}}_{\sf c}(\ell) \bar {\bf{f}}}{\bar {\bf{f}}^{\sf H} {\bf{B}}_{\sf c}(\ell) \bar {\bf{f}} } \right)\right)} \frac{{\bf{B}}_{\sf c}(k)}{\bar {\bf{f}}^{\sf H} {\bf{B}}_{\sf c}(k) \bar {\bf{f}}} + \frac{{\bf{B}}_k}{\bar {\bf{f}}^{\sf H} {\bf{B}}_k \bar {\bf{f}}} \right] , \label{eq:lem_B_kkt}
\end{align}
with
\begin{align} \label{eq:lem_lambda}
    \lambda(\bar {\bf{f}}) =& \left\{\frac{1}{K  } \sum_{k = 1}^{K} \exp\left( \log_2
    \left(\frac{\bar {\bf{f}}^{\sf H} {\bf{A}}_{\sf c}(k) \bar {\bf{f}}}{\bar {\bf{f}}^{\sf H} {\bf{B}}_{\sf c}(k) {\bf{f}}} \right)^{-\frac{1}{\alpha }} \right)\right\}^{-\frac{\alpha}{\log_2 e}}   \times  \prod_{k = 1}^{K} \left(\frac{\bar {\bf{f}}^{\sf H} {\bf{A}}_k \bar {\bf{f}}}{\bar {\bf{f}}^{\sf H} {\bf{B}}_k \bar {\bf{f}}} \right) = \frac{\lambda_{\sf num} (\bar {\bf{f}})}{\lambda_{\sf den} (\bar {\bf{f}})}.
\end{align}
\end{lemma}
\begin{proof}
See Appendix \ref{proof:lem1}.
\end{proof}

Now we interpret the derived optimality condition \eqref{eq:lem_kkt}. We first observe if a precoding vector $\bar {\bf{f}}$ satisfies the condition \eqref{eq:lem_kkt}, then it also satisfies the first-order optimality condition, which means that the corresponding $\bar {\bf{f}}$ is a stationary point of the problem \eqref{eq:problem_new} whose gradient is zero. If the problem \eqref{eq:problem_new} has multiple stationary points, it is possible to exist multiple $\bar {\bf{f}}$ satisfying \eqref{eq:lem_kkt}. 
Next, we see that \eqref{eq:lem_kkt} is presented as a form of the eigenvector problem for the matrix ${\bf{B}}_{\sf KKT}^{-1} (\bar {\bf{f}}^{}){\bf{A}}_{\sf KKT}(\bar {\bf{f}}^{})$. More rigorously, \eqref{eq:lem_kkt} is cast as a NEPv \cite{cai:siam:18}.
As described in \cite{cai:siam:18}, NEPv is a generalized version of an eigenvalue problem, in that a matrix can be changed depending on an eigenvector in a nonlinear fashion. 
In our case, the matrix ${\bf{B}}_{\sf KKT}^{-1} (\bar {\bf{f}}^{}){\bf{A}}_{\sf KKT}(\bar {\bf{f}}^{})$ is a nonlinear function of the eigenvector $\bar {\bf{f}}$. 
Crucially, in the formulated NEPv \eqref{eq:lem_kkt}, the eigenvalue $\lambda(\bar {\bf{f}})$ is equivalent to the objective function of the problem \eqref{eq:problem_new}. Accordingly, if we find the leading eigenvector of the NEPv \eqref{eq:lem_kkt}, then it maximizes the objective function among multiple eigenvectors. 
Eventually, since \eqref{eq:lem_kkt} holds for any eigenvector, finding the leading eigenvector of the NEPv \eqref{eq:lem_kkt} is equivalent to finding the local optimal point that maximizes the objective function of \eqref{eq:problem_new} and has zero gradient. This leads to the following proposition.

\begin{proposition} \label{prop:optimal}
Denoting that the local optimal point for the problem \eqref{eq:problem_new} as $\bar {\bf{f}}^{\star }$, $\bar {\bf{f}}^{\star}$ is the leading eigenvector of ${\bf{B}}_{\sf KKT}^{-1} (\bar {\bf{f}}^{\star}){\bf{A}}_{\sf KKT}(\bar {\bf{f}}^{\star})$ satisfying
\begin{align}
    \label{eq:KKT}
 {\bf{B}}_{\sf KKT}^{-1}(\bar {\bf{f}}^{\star}){\bf{A}}_{\sf KKT}(\bar {\bf{f}}^{\star}) \bar {\bf{f}}^{\star} = \lambda^{\star} \bar {\bf{f}}^{\star},  
\end{align}
where $\lambda^{\star}$ is the corresponding eigenvalue. 
\end{proposition}


Finding $\bar {\bf{f}}^{\star}$ is, however, not straightforward due to the intertwined nature of the problem. 
In the next subsection, we propose a novel method called GPI-RS. GPI-RS is able to obtain the leading eigenvector of the matrix ${\bf{B}}_{\sf KKT}(\bar {\bf{f}}^{})^{-1}{\bf{A}}_{\sf KKT}(\bar {\bf{f}}^{})$ in a computationally efficient fashion.

\begin{algorithm} [t]
\caption{GPI-RS} \label{alg:main} 
\begin{algorithmic} 
\State {\bf{initialize}}: $\bar {\bf{f}}_{(0)} = {\text{MRT}}$
\State Set the iteration count $t = 1$.
\WHILE {$\left\|\bar {\bf{f}}_{(t)} - \bar {\bf{f}}_{(t-1)} \right\| > \epsilon$}
\State Construct the matrices $ {\bf{A}}_{\sf KKT} (\bar {\bf{f}}_{(t-1)})$ and $ {\bf{B}}_{\sf KKT} (\bar {\bf{f}}_{(t-1)})$ by using \eqref{eq:lem_A_kkt} and \eqref{eq:lem_B_kkt}. 
\State Update $\bar {\bf{f}}_{(t)} \leftarrow \frac{{\bf{B}}_{\sf KKT} (\bar {\bf{f}}_{(t-1)})^{-1} {\bf{A}}_{\sf KKT} (\bar {\bf{f}}_{(t-1)}) \bar {\bf{f}}_{(t-1)}} {\| {\bf{B}}_{\sf KKT} (\bar {\bf{f}}_{(t-1)})^{-1} {\bf{A}}_{\sf KKT} (\bar {\bf{f}}_{(t-1)}) \bar {\bf{f}}_{(t-1)} \|}$.
\State $t \leftarrow t+1$. 
\ENDWHILE
\end{algorithmic}
\end{algorithm}

\subsection{Generalized Power Iteration for Rate-Splitting}
The basic process of the proposed GPI-RS follows that of the conventio power iteration. 
Given $\bar {\bf{f}}_{(t-1)}$ obtained in the $(t-1)$th iteration, we construct the matrices $ {\bf{B}}_{\sf KKT} (\bar {\bf{f}}_{(t-1)})$ and ${\bf{A}}_{\sf KKT} (\bar {\bf{f}}_{(t-1)})$ using \eqref{eq:lem_A_kkt} and \eqref{eq:lem_B_kkt}. Then, we update the precoding vector for the current iteration as 
\begin{align}
    \label{eq:PowerIter}
    \bar {\bf{f}}_{(t)} \leftarrow \frac{{\bf{B}}_{\sf KKT}^{-1} (\bar {\bf{f}}_{(t-1)}) {\bf{A}}_{\sf KKT} (\bar {\bf{f}}_{(t-1)}) \bar {\bf{f}}_{(t-1)}}{\| {\bf{B}}_{\sf KKT}^{-1} (\bar {\bf{f}}_{(t-1)}) {\bf{A}}_{\sf KKT} (\bar {\bf{f}}_{(t-1)}) \bar {\bf{f}}_{(t-1)} \|}.
\end{align}
We repeat this process until the convergence criterion is met. In this paper, we use $\left\|\bar {\bf{f}}_{(t)} - \bar {\bf{f}}_{(t-1)} \right\| < \epsilon$ for small enough $\epsilon$. We summarize this process in Algorithm \ref{alg:main}. For an initial point $\bar {\bf{f}}_{(0)}$, we use maximum ratio transmission (MRT), which works well in the later simulations. 

\begin{remark} \normalfont (Joint power control and beamforming) \label{remark:joint}
The proposed GPI-RS identifies the leading eigenvector $\bar {\bf{f}}^{\star}$ that maximizes \eqref{eq:lem_kkt}. Since the vector $\bar {\bf{f}}$ is constructed by stacking all the precoding vectors corresponding to each message, the power allocation and the beamforming direction of each message are jointly identified within the found vector. For example, if $\|\bar {\bf{f}}^{\star}(1:N) \| = 0$, this means $\| {\bf{f}}_{\sf c} \| = 0$ in the obtained solution. Then the common message is not assigned any transmit power, therefore we do not use RSMA and go back to use classical SDMA. Thanks to this feature, the proposed GPI-RS automatically determines the message setups depending on channel conditions; so that there is no need to employ a separate process to determine whether to use a common message. 

\end{remark}

\begin{remark} \normalfont (Algorithm complexity) \label{remark:complexity}
The total computational complexity of the proposed GPI-RS is dominated by the calculation of ${\bf{B}}^{-1}_{\sf KKT}(\bar {\bf{f}})$. The matrix ${\bf{B}}^{-1}_{\sf KKT}(\bar {\bf{f}})$ is the sum of the block-diagonal matrices as presented in \eqref{eq:lem_B_kkt}. Specifically, $K+1$ number of $N \times N$ submatrices are concatenated, so that the total size is $(K+1)N \times (K+1)N$. For this reason, the inverse matrix ${\bf{B}}^{-1}_{\sf KKT}(\bar {\bf{f}})$ is obtained by computing the inverse of each submatrix, and this requires the complexity with the order of $\CMcal{O}(\frac{1}{3} (K+1)N^3)$. This results in that the complexity of the proposed GPI-RS per iteration is with the order of $\CMcal{O}(\frac{1}{3}KN^3)$ when $N$ and $K$ increase with the same order. 
We note that this is substantially small compared to the existing methods. For example, the conventional WMMSE methods based on QCQP \cite{joudeh:16:tcom, joude:twc:17} need the complexity order of $\CMcal{O}((KN)^{3.5})$ {\color{black}{\cite{patil:tcom:18, joudeh:16:tsp}}}. 
Further, the CCCP based method \cite{li:jsac:20} is associated with the complexity order of $\CMcal{O}(N^{6}K^{0.5}2^{3.5K})$. 
In particular, it is noteworthy that the proposed GPI-RS has the linear order complexity with the number of user $K$; which makes the proposed GPI-RS advantageous when there are a large number of users. 
Additionally, our algorithm is easy to implement in practice in that CVX is not needed to use to obtain a solution. 
{\color{black}{We note that the computational complexity of the proposed method can be further reduced by adopting matrix inversion approximation techniques such as Chebyshev iteration \cite{peng:tsp:17} or Neumann series \cite{zhu:icc:15}, or by limiting the maximum number of iterations in the proposed method.}}
\end{remark}

\begin{remark} \normalfont (Selection of the parameter $\alpha$)
Even though small $\alpha$ is desirable since it provides accurate approximation in the LogSumExp technique, using too small $\alpha$ may cause the algorithm to diverge. Analytically identifying the optimal $\alpha$, however, is very challenging. To find the proper $\alpha$ numerically, we can modify the GPI-RS to obtain the smallest $\alpha$ that makes the GPI-RS algorithm converge. Specifically, we start the GPI-RS with a small $\alpha$. If the iteration loop of the GPI-RS does not converge within the predetermined number of iterations, then we enforce to terminate the loop, increase $\alpha$, and newly start the algorithm again. We repeat this process until the algorithm converges before the predetermined number. We can empirically adapt the starting $\alpha$ value and the increasing ratio depending on the system configuration to reduce the algorithm time. 
\end{remark}

{\color{black}{
\begin{remark} \normalfont (Principle of GPI-RS)
We explain the principle of the GPI-RS algorithm through the conventional power iteration. In the conventional power iteration, we obtain the leading eigenvector of a matrix ${\bf{M}} \in \mathbb{C}^{n \times n}$ by iteratively calculating ${\bf{q}}_{(t+1)} = \frac{{\bf{M}}^{t} {\bf{q}}_{(0)}}{\| {\bf{M}}^{t} {\bf{q}}_{(0)} \| }$. A rationale that the power iteration converges to the leading eigenvector of ${\bf{M}}$ is as follows. Since a set of eigenvectors form a set of basis, we can represent ${\bf{q}}_{(0)} = \sum_{i = 1}^{n} \alpha_i {\bf{x}}_i$, where ${\bf{x}}_i$ indicates the $i$-th eigenvector and $\alpha_i$ is the corresponding weight. Denoting that $\lambda_i$ is the $i$-th eigenvalue that $|\lambda_1| > |\lambda_2| \ge \cdots \ge |\lambda_n|$, we use ${\bf{M}} = \sum_{i = 1}^{n} \lambda_i {\bf{x}}_i $ to derive 
\begin{align}
    {\bf{M}} {\bf{q}}_{(t)} 
    &= \sum_{i = 1}^{n} \alpha_i \lambda_i^t {\bf{x}}_i = \alpha_1 \lambda_1^t \left( {\bf{x}}_1 + \underbrace{ \sum_{i = 2}^{n} \frac{\alpha_i}{\alpha_1} \left(\frac{\lambda_i}{\lambda_1} \right)^t {\bf{x}}_i}_{(a)} \right).
 \end{align}
As $t \rightarrow \infty$, (a) vanishes; thereby the remaining term converges to the leading eigenvector ${\bf{x}}_1$. 

As presented in Proposition \ref{prop:optimal}, our problem generalizes a conventional eigenvalue problem by considering an eigenvector-dependent matrix ${\bf{M}}({\bf{x}})$, known as NEPv \cite{cai:siam:18}. That is to say, we aim to identify ${\bf{x}}_1$ that fulfills ${\bf{M}}({\bf{x}}_1) {\bf{x}}_1 = \lambda_1 {\bf{x}}_1$, where $\lambda_1$ is the maximum eigenvalue and ${\bf{M}}({\bf{x}}) = {\bf{B}}_{\sf KKT}( {\bf{x}}^{})^{-1}{\bf{A}}_{\sf KKT}( {\bf{x}}^{})$ in our case. For convenience, denote ${\bf{M}}({\bf{x}}_i) {\bf{x}}_i = g({\bf{x}}_i)$. Then with an arbitrary vector ${\bf{x}}$, the Taylor expansion at ${\bf{x}}_1$ for $g({\bf{x}})$ leads to
\begin{align}
      g({\bf{x}})^{\sf H}{\bf{x}}_i = g({\bf{x}}_1)^{\sf H} {\bf{x}}_i + ({\bf{x}} - {\bf{x}}_1)^{\sf H} \nabla g({\bf{x}}_1) {\bf{x}}_i + o(\| {\bf{x}} - {\bf{x}}_1 \|). 
\end{align}
On one hand, we have
\begin{align}
    \left(g({\bf{x}})^{\sf H}{\bf{x}}_1\right)^2 = \left(\lambda_1 + o(\|{\bf{x}} - {\bf{x}}_1 \|)\right)^2
\end{align}
due to the fact that ${\bf{M}}({\bf{x}}_1) {\bf{x}}_1 = \lambda_1 {\bf{x}}_1$. On the other hand, assuming that $\{{\bf{v}}_1, \cdots, {\bf{v}}_n \}$ is a set of orthonormal basis where ${\bf{v}}_1 = {\bf{x}}_1$, we also have 
\begin{align}
    \sum_{i = 2}^{n} \left( g({\bf{x}})^{\sf H}{\bf{v}}_i \right)^2  & \le  \sum_{i = 2}^{n} \left[ \lambda_i^2({\bf{x}}^{\sf H} {\bf{v}}_i)^2 + 2\lambda_i ({\bf{x}}^{\sf H} {\bf{v}}_i) o(\| {\bf{x}} - {\bf{x}}_1 \|) + o(\| {\bf{x}} - {\bf{x}}_1 \|)^2 \right]
    \\
    & \le \left(\lambda_2 \|{\bf{x}} - {\bf{x}}_1 \| + o(\| {\bf{x}} - {\bf{x}}_1 \|)\right)^2
\end{align}
Under a premise that $|\lambda_1| > |\lambda_2| \ge |\lambda_i|, \forall i \neq 1, 2$, by iteratively projecting ${\bf{x}}$ onto ${\bf{M}}({\bf{x}})$ with the GPI-RS algorithm, each component corresponding to non-leading eigenvectors ${\bf{x}}_2, \cdots, {\bf{x}}_n$ vanishes. Accordingly, the GPI-RS converges to the leading eigenvector ${\bf{x}}_1$. 

\end{remark}
}}


\section{Generalization to Multiple-layer Rate Splitting}
This section discusses how to generalize the proposed framework to multiple-layer RSMA. As mentioned in \cite{dai:16:twc}, if groups of users are located within multiple clusters, it is beneficial to exploit a partial common message that includes the messages of a subset of the users. Employing multiple-layer RSMA, the transmit signal ${\bf{x}}$ is given by
\begin{align} \label{eq:signal_model_partial}
    {\bf{x}} = {\bf{f}}_{\sf c}s_{\sf c} + \sum_{i =1}^{G} {\bf{f}}_{{\sf c}, \CMcal{K}_i} s_{{\sf c}, \CMcal{K}_i} + \sum_{k = 1}^{K} {\bf{f}}_k s_k,
\end{align}
where the partial common message $s_{{\sf c}, \CMcal{K}_i}$ is decoded for the users included in $\CMcal{K}_i \subset \CMcal{K}$. We denote $|\CMcal{K}_i| = K_i$. 
If $G = 0$, i.e., there is no partial common message, then \eqref{eq:signal_model_partial} reduces to the $1$-layer RSMA \eqref{eq:transmit_signal}. Assuming $\CMcal{K}_i \cap \CMcal{K}_j = \emptyset$, user $k$ for $ k \in \CMcal{K}_i$ decodes the messages with the following order: $s_{\sf c} \rightarrow s_{{\sf c}, \CMcal{K}_i} \rightarrow s_k$. 
To guarantee that the partial common message $s_{{\sf c}, \CMcal{K}_i}$ is successfully decoded for the users in $\CMcal{K}_i$, the information rate of $s_{{\sf c}, \CMcal{K}_i}$ is determined as 
\begin{align}
    R_{{\sf c},\CMcal{K}_i} = \mathop{\min}_{k \in \CMcal{K}_i} \left\{\mathbb{E}_{\{\hat {\bf{h}}_k \}} \left[ R_{{\sf c}, \CMcal{K}_i}^{\sf ins.} (k) \right] \right\}.
\end{align}
By using the same technique presented in Section \ref{sec:perf_ch}, we derive a lower bound as
\begin{align}
    R_{{\sf c}, \CMcal{K}_i} \ge \mathbb{E}_{\{\hat {\bf{h}}_{k \in \CMcal{K}_i}\}} \left[ \mathop{\min}_{k \in \CMcal{K}_i} \left\{ \bar R_{{\sf c}, \CMcal{K}_i}^{\sf ins.}(k)\right\} \right],
\end{align}
where
\begin{align} \label{eq:ins_se_lb_pcomm}
    \bar R_{{\sf c}, \CMcal{K}_i}^{\sf ins.}(k) = \log_2 \left( 1 + \frac{|\hat {\bf{h}}_{k}^{\sf H} {\bf{f}}_{{\sf c}, \CMcal{K}_i}|^2}{\sum_{j = 1, j \neq i}^{G} |\hat {\bf{h}}_{k}^{\sf H} {\bf{f}}_{{\sf c}, \CMcal{K}_j}|^2 +  \sum_{\ell = 1}^{K} |\hat {\bf{h}}_{k}^{\sf H} {\bf{f}}_\ell|^2 + \sum_{j = 1}^{G}{\bf{f}}_{{\sf c}, \CMcal{K}_j}^{\sf H} {\bf{\Phi}}_k {\bf{f}}_{{\sf c}, \CMcal{K}_j} + \sum_{\ell = 1}^{K} {\bf{f}}_{\ell}^{\sf H} {\bf{\Phi}}_k {\bf{f}}_{\ell}  +  \frac{\sigma^2}{P}} \right).
\end{align}
We observe that the SINR in \eqref{eq:ins_se_lb_pcomm} does not have interference from the common message $s_{\sf c}$ since we assume that the common message is already decoded and eliminated via SIC. Similar to this, we also characterize a lower bound on the instantaneous spectral efficiency for the common message and the private message as follows: 
\begin{align} \label{eq:ins_se_lb_pcomm_others}
    &\bar R_{\sf c}^{\sf ins.}(k) \nonumber \\
    & = \log_2 \left( 1 + \frac{|\hat {\bf{h}}_{k}^{\sf H} {\bf{f}}_{{\sf c}}|^2}{\sum_{j = 1}^{G} |\hat {\bf{h}}_{k}^{\sf H} {\bf{f}}_{{\sf c}, \CMcal{K}_j}|^2 +  \sum_{\ell = 1}^{K} |\hat {\bf{h}}_{k}^{\sf H} {\bf{f}}_\ell|^2 + {\bf{f}}_{\sf c}^{\sf H}{\bf{\Phi}}_k {\bf{f}}_{\sf c} + \sum_{j = 1}^{G}{\bf{f}}_{{\sf c}, \CMcal{K}_j}^{\sf H} {\bf{\Phi}}_k {\bf{f}}_{{\sf c}, \CMcal{K}_j} + \sum_{\ell = 1}^{K} {\bf{f}}_{\ell}^{\sf H} {\bf{\Phi}}_k {\bf{f}}_{\ell}  +  \frac{\sigma^2}{P}} \right),\\
    &\bar R_k = \log_2 \left( 1 + \frac{|\hat {\bf{h}}_{k}^{\sf H} {\bf{f}}_{k}|^2}{\sum_{j = 1, j \neq i}^{G} |\hat {\bf{h}}_{k}^{\sf H} {\bf{f}}_{{\sf c}, \CMcal{K}_j}|^2 +  \sum_{\ell = 1, \ell \neq k}^{K} |\hat {\bf{h}}_{k}^{\sf H} {\bf{f}}_\ell|^2 + \sum_{j = 1, j\neq i}^{G}{\bf{f}}_{{\sf c}, \CMcal{K}_j}^{\sf H} {\bf{\Phi}}_k {\bf{f}}_{{\sf c}, \CMcal{K}_j} + \sum_{\ell = 1}^{K} {\bf{f}}_{\ell}^{\sf H} {\bf{\Phi}}_k {\bf{f}}_{\ell}  +  \frac{\sigma^2}{P}} \right).
\end{align}
Accordingly, we formulate the sum spectral efficiency maximization problem for multiple-layer RSMA as 
\begin{align} \label{eq:lb_problem_pcomm}
\mathop{{\text{maximize}}}_{{\bf{f}}_{\sf c}, {\bf{f}}_{{\sf c}, \CMcal{K}_1}, \cdots, {\bf{f}}_{{\sf c}, \CMcal{K}_G}, {\bf{f}}_1, \cdots,{\bf{f}}_K}& \;\; \min_{k \in \CMcal{K }} \{\bar R_{\sf c}^{\sf ins.}(k)\} + \sum_{i = 1}^{G} \min_{k \in \CMcal{K }_i} \{\bar R_{{\sf c},\CMcal{K}_i}^{\sf ins.}(k)\}  + \sum_{k = 1}^{K} \bar R_{k}^{\sf ins.} \\
{\text{subject to}} & \;\;  \left\| {\bf{f}}_{\sf c} \right\|^2 + \sum_{k = 1}^{K} \left\| {\bf{f}}_k \right\|^2  \le 1 .\label{eq:lb_problem_constraint}
\end{align}
To obtain a solution of \eqref{eq:lb_problem_pcomm}, we reformulate \eqref{eq:lb_problem_pcomm} by following our approach used in the $1$-layer RSMA. We first approximate $ \mathop{\min}_{k \in \CMcal{K}_i} \left\{ \bar R_{{\sf c}, \CMcal{K}_i}^{\sf ins.}(k)\right\} $ by using the LogSumExp technique as
\begin{align} \label{eq:approximation_R}
\min_{k \in \CMcal{K }_i} \{\bar R_{{\sf c}, \CMcal{K}_i}^{\sf ins.}(k)\} \approx  -\alpha \log \left(\frac{1}{K_i} \sum_{k \in \CMcal{K}_i} \exp\left( \frac{\bar R_{{\sf c}, \CMcal{K}_i}^{\sf ins.}(k)}{-\alpha}\right) \right).
\end{align}
By following our approach, we define a high dimensional precoding vector $\bar {\bf{f}}$ as
\begin{align} \label{eq:stack_pcomm}
\bar {\bf{f}} = [{\bf{f}}_{\sf c}^{\sf T}, {\bf{f}}_{{\sf c}, \CMcal{K}_1}^{\sf T}, \cdots, {\bf{f}}_{{\sf c}, \CMcal{K}_G}^{\sf T},  {\bf{f}}_1^{\sf T}, \cdots, {\bf{f}}_K^{\sf T}]^{\sf T} \in \mathbb{C}^{N(K+G+1) \times 1}.
\end{align}
The higher dimensional precoding vector $\bar {\bf{f}}$ in \eqref{eq:stack_pcomm} allows to represent the spectral efficiency expression of the partial common message $s_{{\sf c}, \CMcal{K}_i}$ as a Rayleigh quotient form. 
\begin{align} \label{eq:rewrite_commonrate_partial}
\bar R_{{\sf c}, \CMcal{K}_i}^{\sf ins.}(k) = \log_2 \left( \frac{\bar {\bf{f}}^{\sf H} {\bf{A}}_{{\sf c}, \CMcal{K}_i} (k) \bar {\bf{f}}}{\bar {\bf{f}}^{\sf H} {\bf{B}}_{{\sf c}, \CMcal{K}_i}(k)  \bar {\bf{f}} } \right),
\end{align}
where
\begin{align}
&{\bf{A}}_{{\sf c},\CMcal{K}_i}(k) =  {\rm blkdiag} \left({\bf{0}}, (\hat {\bf{h}}_k \hat {\bf{h}}_k^{\sf H} + {\bf{\Phi}}_k) , \cdots , (\hat {\bf{h}}_k \hat {\bf{h}}_k^{\sf H} + {\bf{\Phi}}_k) \right) + {\bf{I}}_{N(K+G+1)} \frac{ \sigma^2}{P}, \\
& {\bf{B}}_{{\sf c},\CMcal{K}_i}(k) = {\bf{A}}_{{\sf c},\CMcal{K}_i}(k) - {\rm blkdiag} \left({\bf{0}}, \cdots, \underbrace{{\hat {\bf{h}}_k \hat {\bf{h}}_k^{\sf H}}}_{(1+i)\text{th block}}, \cdots, {\bf{0}} \right) .
\end{align}
We also represent $\bar R_{\sf c}^{\sf ins.}(k)$ and $\bar R_k^{\sf ins.}$ adequately by considering the partial common message. 
With this representation, the problem \eqref{eq:lb_problem_pcomm} is converted to
\begin{align} \label{eq:problem_new_pcomm}
\mathop{{\text{maximize}}}_{\bar {\bf{f}}} \;\;  
 &\log \left(\frac{1}{K}  \sum_{k = 1}^{K}  \exp \left(\log_2 \left( \frac{\bar {\bf{f}}^{\sf H} {\bf{A}}_{\sf c} (k) \bar {\bf{f}}}{\bar {\bf{f}}^{\sf H} {\bf{B}}_{\sf c}(k)  \bar {\bf{f}} }  \right)^{-\frac{1}{\alpha}} \right)  \right)^{-\alpha} \nonumber \\ & +\sum_{i = 1}^{G} \log \left(\frac{1}{K_i} \sum_{k \in \CMcal{K}_i} \exp \left(\log_2 \left( \frac{\bar {\bf{f}}^{\sf H} {\bf{A}}_{{\sf c}, \CMcal{K}_i}(k) \bar {\bf{f}}}{\bar {\bf{f}}^{\sf H} {\bf{B}}_{{\sf c}, \CMcal{K}_i}(k) \bar {\bf{f}}}\right) \right)^{-\frac{1}{\alpha}} \right)^{-\alpha}  +  \sum_{k = 1}^{K}  \log_2 \left( \frac{\bar {\bf{f}}^{\sf H} {\bf{A}}_k \bar {\bf{f}}}{\bar {\bf{f}}^{\sf H} {\bf{B}}_k \bar {\bf{f}}}   \right) \\
\label{eq:constraint_new}
{\text{subject to}} & \;\; \left\| \bar {\bf{f}} \right\|^2 = 1.
\end{align}
To apply the GPI-RS algorithm for \eqref{eq:problem_new_pcomm}, we derive the optimality condition for \eqref{eq:problem_new_pcomm} in the following corollary. 
\begin{corollary} \label{coro:pcomm}
 The first-order optimality condition of the optimization problem \eqref{eq:problem_new_pcomm} is satisfied if the following holds:
\begin{align} \label{eq:lem_kkt_pcomm}
{\bf{B}}_{\sf KKT}^{-1} (\bar {\bf{f}}){\bf{A}}_{\sf KKT} (\bar {\bf{f}}) \bar {\bf{f}} = \lambda (\bar {\bf{f}})  \bar {\bf{f}},
\end{align} 
where 
\begin{align} 
    &{\bf{A}}_{\sf KKT}(\bar {\bf{f}}) =  \lambda_{\sf num} (\bar {\bf{f}}) \times \left[ \sum_{k = 1}^{K}  \left\{ \frac{\exp\left( \frac{1}{-\alpha}  \log_2\left(\frac{\bar {\bf{f}}^{\sf H} {\bf{A}}_{\sf c}(k) \bar {\bf{f}}}{\bar {\bf{f}}^{\sf H} {\bf{B}}_{\sf c}(k) \bar {\bf{f}} } \right) \right)}{\sum_{\ell = 1}^{K} \exp\left(\frac{1}{-\alpha} \log_2\left(\frac{\bar {\bf{f}}^{\sf H} {\bf{A}}_{\sf c}(\ell) \bar {\bf{f}}}{\bar {\bf{f}}^{\sf H} {\bf{B}}_{\sf c}(\ell) \bar {\bf{f}} } \right)\right)} \frac{{\bf{A}}_{\sf c}(k)}{\bar {\bf{f}}^{\sf H} {\bf{A}}_{\sf c}(k) \bar {\bf{f}}} + \frac{{\bf{A}}_k}{\bar {\bf{f}}^{\sf H} {\bf{A}}_k \bar {\bf{f}}} \right\}\right. \nonumber \\
    & + \left. \sum_{i = 1}^{G} \left\{ \sum_{k \in \CMcal{K}_i} \left(\frac{\exp \left( \frac{1}{-\alpha} \frac{\bar {\bf{f}}^{\sf H} {\bf{A}}_{{\sf c}, \CMcal{K}_i}(k) \bar {\bf{f}}}{\bar {\bf{f}}^{\sf H} {\bf{B}}_{{\sf c}, \CMcal{K}_i}(k) \bar {\bf{f}}} \right)}{\sum_{j \in \CMcal{K}_i} \exp \left(\frac{1}{-\alpha} \log_2 \left( \frac{\bar {\bf{f}}^{\sf H} {\bf{A}}_{{\sf c}, \CMcal{K}_i}(j) \bar {\bf{f}}}{\bar {\bf{f}}^{\sf H} {\bf{B}}_{{\sf c}, \CMcal{K}_i}(j) \bar {\bf{f}}} \right) \right) } \frac{{\bf{A}}_{{\sf c}, \CMcal{K}_i} (k) }{\bar {\bf{f}}^{\sf H} {\bf{A}}_{{\sf c}, \CMcal{K}_i} (k) \bar {\bf{f}} } \right) \right\}\right], \label{eq:lem_A_kkt_pcomm} 
\end{align}
\begin{align}
    &{\bf{B}}_{\sf KKT}(\bar {\bf{f}}) = \lambda_{\sf den} (\bar {\bf{f}}) \times \left[ \sum_{k = 1}^{K}  \left\{ \frac{\exp\left( \frac{1}{-\alpha}  \log_2\left(\frac{\bar {\bf{f}}^{\sf H} {\bf{A}}_{\sf c}(k) \bar {\bf{f}}}{\bar {\bf{f}}^{\sf H} {\bf{B}}_{\sf c}(k) \bar {\bf{f}} } \right) \right)}{\sum_{\ell = 1}^{K} \exp\left(\frac{1}{-\alpha} \log_2\left(\frac{\bar {\bf{f}}^{\sf H} {\bf{A}}_{\sf c}(\ell) \bar {\bf{f}}}{\bar {\bf{f}}^{\sf H} {\bf{B}}_{\sf c}(\ell) \bar {\bf{f}} } \right)\right)} \frac{{\bf{B}}_{\sf c}(k)}{\bar {\bf{f}}^{\sf H} {\bf{B}}_{\sf c}(k) \bar {\bf{f}}} + \frac{{\bf{B}}_k}{\bar {\bf{f}}^{\sf H} {\bf{B}}_k \bar {\bf{f}}} \right\} \right. \nonumber \\
    & + \left. \sum_{i = 1}^{G} \left\{ \sum_{k \in \CMcal{K}_i} \left(\frac{\exp \left( \frac{1}{-\alpha} \frac{\bar {\bf{f}}^{\sf H} {\bf{A}}_{{\sf c}, \CMcal{K}_i}(k) \bar {\bf{f}}}{\bar {\bf{f}}^{\sf H} {\bf{B}}_{{\sf c}, \CMcal{K}_i}(k) \bar {\bf{f}}} \right)}{\sum_{j \in \CMcal{K}_i} \exp \left(\frac{1}{-\alpha} \log_2 \left( \frac{\bar {\bf{f}}^{\sf H} {\bf{A}}_{{\sf c}, \CMcal{K}_i}(j) \bar {\bf{f}}}{\bar {\bf{f}}^{\sf H} {\bf{B}}_{{\sf c}, \CMcal{K}_i}(j) \bar {\bf{f}}} \right) \right) } \frac{{\bf{B}}_{{\sf c}, \CMcal{K}_i} (k) }{\bar {\bf{f}}^{\sf H} {\bf{B}}_{{\sf c}, \CMcal{K}_i} (k) \bar {\bf{f}} } \right) \right\}\right]
    , \label{eq:lem_B_kkt_pcomm}
\end{align}
with
\begin{align} \label{eq:lem_lambda_pcomm}
    \lambda(\bar {\bf{f}}) =& \left\{\frac{1}{K  } \sum_{k = 1}^{K} \exp\left( \log_2
    \left(\frac{\bar {\bf{f}}^{\sf H} {\bf{A}}_{\sf c}(k) \bar {\bf{f}}}{\bar {\bf{f}}^{\sf H} {\bf{B}}_{\sf c}(k) {\bf{f}}} \right)^{-\frac{1}{\alpha }} \right)\right\}^{-\frac{\alpha}{\log_2 e}}   \times  \prod_{k = 1}^{K} \left(\frac{\bar {\bf{f}}^{\sf H} {\bf{A}}_k \bar {\bf{f}}}{\bar {\bf{f}}^{\sf H} {\bf{B}}_k \bar {\bf{f}}} \right) \nonumber \\
    & \times \prod_{i = 1}^{G} \left\{\frac{1}{K_i} \sum_{k \in \CMcal{K}_i} \exp \left(\log_2 \left( \frac{\bar {\bf{f}}^{\sf H} {\bf{A}}_{{\sf c}, \CMcal{K}_i}(k) \bar {\bf{f}}}{\bar {\bf{f}}^{\sf H} {\bf{B}}_{{\sf c}, \CMcal{K}_i}(k) \bar {\bf{f}}}\right) \right)^{-\frac{1}{\alpha}} \right\}^{-\frac{\alpha}{\log_2 e}} = \frac{\lambda_{\sf num} (\bar {\bf{f}})}{\lambda_{\sf den} (\bar {\bf{f}})} .
\end{align}
\end{corollary}
\begin{proof}
It can be proven by extending Lemma \ref{lem:main}. 
\end{proof}
With Corollary \ref{coro:pcomm}, we apply the GPI-RS described in Algorithm \ref{alg:main} by using \eqref{eq:lem_A_kkt_pcomm} and \eqref{eq:lem_B_kkt_pcomm} instead of \eqref{eq:lem_A_kkt} and \eqref{eq:lem_B_kkt}. 



\section{Numerical Results}
\label{sec:numerical}
In this section, we evaluate the sum spectral efficiency performance to demonstrate the proposed GPI-RS. For the baseline methods, we consider the followings:
\begin{itemize}
    \item {{{\textbf{MRT}}}}: The precoding vectors is designed by matching the estimated channel vector. Specifically, we have 
    ${\bf{f}}_k = \hat {\bf{h}}_k$, $k \in \CMcal{K},\; \mathrm{and} \;\;{\bf{f}}_{\sf c} = \textbf{0}$. 
    \item {{{\textbf{RZF}}}} : The precoding vectors are designed by following the ZF rule, while regularizing it depending on SNR:
    \begin{align}
        {\bf{f}}_k = \left(\hat {\bf{H}} \hat {\bf{H}}^{\sf H} + {\bf{I}} \frac{\sigma^2}{P} \right)^{-1} \hat {\bf{h}}_k^{\sf H}, \; k \in \CMcal{K}, \; \mathrm{and} \;\;{\bf{f}}_{\sf c} = \textbf{0}.
    \end{align}
    As SNR goes to infinity, RZF becomes equal to ZF. 
    \item {{{\textbf{Sum SE Max with no RS}}}}: In this method, we use the method proposed in \cite{choi:twc:20} to maximize the sum spectral efficiency using classical SDMA without considering RSMA. 
    Note that we do not incorporate the CSIT estimation error into this method, i.e., we treat the estimated channel vector as the true channel in this method. 
    \item {{{\textbf{WMMSE-SAA}}}}: This case indicates the WMMSE method with the SAA technique \cite{joudeh:16:tcom}. We use $1000$ samples for the SAA technique. A detailed description of this approach is presented in Section \ref{sec:preliminary}. 
\end{itemize}

In what follows, we present the simulation results. 

{\textbf{Ergodic Sum Spectral Efficiency per SNR}}: First, we compare the ergodic sum spectral efficiency of the proposed GPI-RS and the other baseline methods. 
The basic simulation setups are explained in the caption of Fig.~\ref{fig:compare}. For updating $\alpha$, we set the initial $\alpha$ value as $0.1$ if ${\sf{SNR}} < 15 {\text{dB}}$, and $0.5$ for the rest of the cases. We note that this initial setup is designed empirically. If the GPI-RS loop is not terminated within $50$ iterations, we increase $\alpha$ by $0.5$ and repeat the algorithm. 
As shown in Fig.~\ref{fig:compare}, the proposed GPI-RS provides meaningful spectral efficiency gains over the baseline methods in both of the $6\times4$ and the $12\times8$ cases. In particular, compared to the WMMSE-SAA method at ${\sf{SNR}} = 40 {\text{dB}}$, the GPI-RS obtains around $24\%$ gains in the $6\times4$ case and $19\%$ gains in the $12\times8$ case. 
We observe that considerable gains are achieved in both cases by using the proposed method. The rationales of the performance gains are two folds. First, the GPI-RS can reach the best local optimal point by NEPv principle, while the WMMSE-SAA approach cannot guarantee the best local optimum. Second, we incorporate the CSIT estimation error into our performance characterization in a rigorous way, while the WMMSE-SAA relies on randomly generated samples. 
These two features bring the gains of the proposed GPI-RS method. 
{\color{black}{Especially, the second point sheds light on a reason why the proposed GPI-RS obtains more significant performance gains in the high SNR regime. In the high SNR regime, the performance is mainly determined by the interference induced from the inaccurate CSIT as the noise becomes negligible. For this reason, to achieve high spectral efficiency performance, rigorous treatment of CSIT error is required. The SAA, used in the WMMSE-SAA, relies on randomly generated samples instead of incorporating the CSIT estimation error effects into the spectral efficiency performance. 
Compared to this, the GPI-RS specifically incorporates the CSIT estimation error effects into its design by deriving a rigorous lower bound. This difference causes more significant performance gaps in the high SNR regime. 
}}{\color{black}{ Additionally, by comparing the sum SE max with no RS case, we see that the spectral efficiency gains of the proposed GPI-RS increase as SNR increases. This indicates that the common message rate portion increases in the high SNR regime, which matches our intuition on RSMA. 
}}

{\color{black}{In Fig.~\ref{fig:large}, we also illustrate the ergodic sum spectral efficiency  by assuming the the $32 \times 16$ case. 
As in the above case, we also observe that the GPI-RS provides significant gains over the baseline methods in Fig.~\ref{fig:large}. Specifically, at ${\sf{SNR}} = 40 {\text{dB}}$, the GPI-RS offers about $30 \%$ gains compared to the sum spectral efficiency maximization method without RSMA. This demonstrates that the proposed GPI-RS works well in the regimes of large BS antennas and users. 
}}

\begin{figure}[!t]
\centering
$\begin{array}{cc}  
{\resizebox{0.45\columnwidth}{!}
{\includegraphics{ 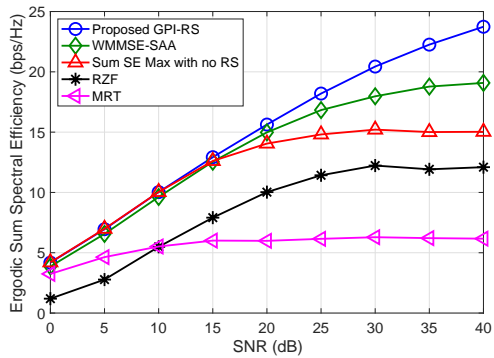}}} \\
\mbox{(a)} \\
{\resizebox{0.45\columnwidth}{!}
{\includegraphics{ 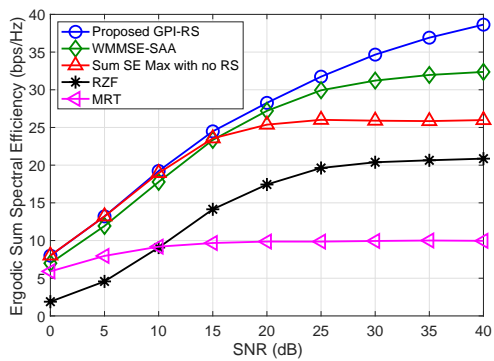}}}  \\  
\mbox{(b)}
\end{array}$
\caption{The sum spectral efficiency comparison per SNR. The simulation setup is: In (a), $N = 6$, $K = 4$, $\tau_{\sf ul} p_{\sf ul} = 4$ and $\sigma^2 = 1$. In (b) $N = 12$, $K = 8$, $\tau_{\sf ul} p_{\sf ul} = 4$ and $\sigma^2 = 1$.
In both cases, each user's location is randomly determined, so that the AoA $\theta_k$ is drawn from the uniform distribution. The angular spread is fixed as $\Delta_k = \pi/6$ for $k \in \CMcal{K}$. 
}
\label{fig:compare}  
\end{figure}

 \begin{figure}[!t]
	\centerline{\resizebox{0.45\columnwidth}{!}{\includegraphics{ 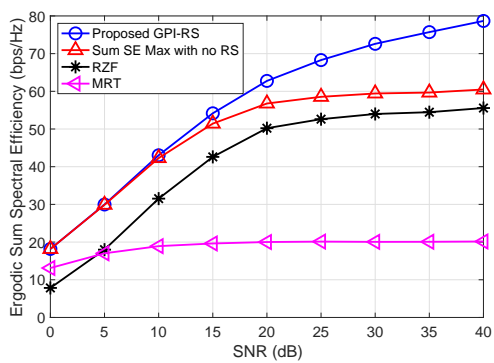}}}     
	\caption{The sum spectral efficiency comparison per SNR assuming $N = 32$, $K = 16$. The other setups are same with Fig.~\ref{fig:compare}. }
 	\label{fig:large}
\end{figure}

{\textbf{Ergodic Sum Spectral Efficiency per CSIT Accuracy}}: Now, we investigate the sum spectral efficiency depending on the CSIT accuracy. We depict the sum spectral efficiency by increasing $\tau_{\sf ul} p_{\sf ul}$ in Fig.~\ref{fig:perErr}. Note that the CSIT estimation accuracy increases as $\tau_{\sf ul} p_{\sf ul}$ increases. In Fig.~\ref{fig:perErr}, we observe that the relative performance gains of the GPI-RS over the WMMSE-SAA increase as the CSIT becomes accurate. For instance, if $\tau_{\sf ul} p_{\sf ul} = 0.1$, the WMMSE-SAA outperforms the GPI-RS by $5\%$, while if $\tau_{\sf ul} p_{\sf ul} = 8$, the GPI-RS outperforms the WMMSE-SAA by $27\%$. This is because, as the CSIT is more accurate, the regularization term of our lower bound also becomes accurate; resulting in that our lower bound becomes tight. Then the CSIT estimation error is suitably reflected into the GPI-RS design. 

 \begin{figure}[!t]
	\centerline{\resizebox{0.45\columnwidth}{!}{\includegraphics{ 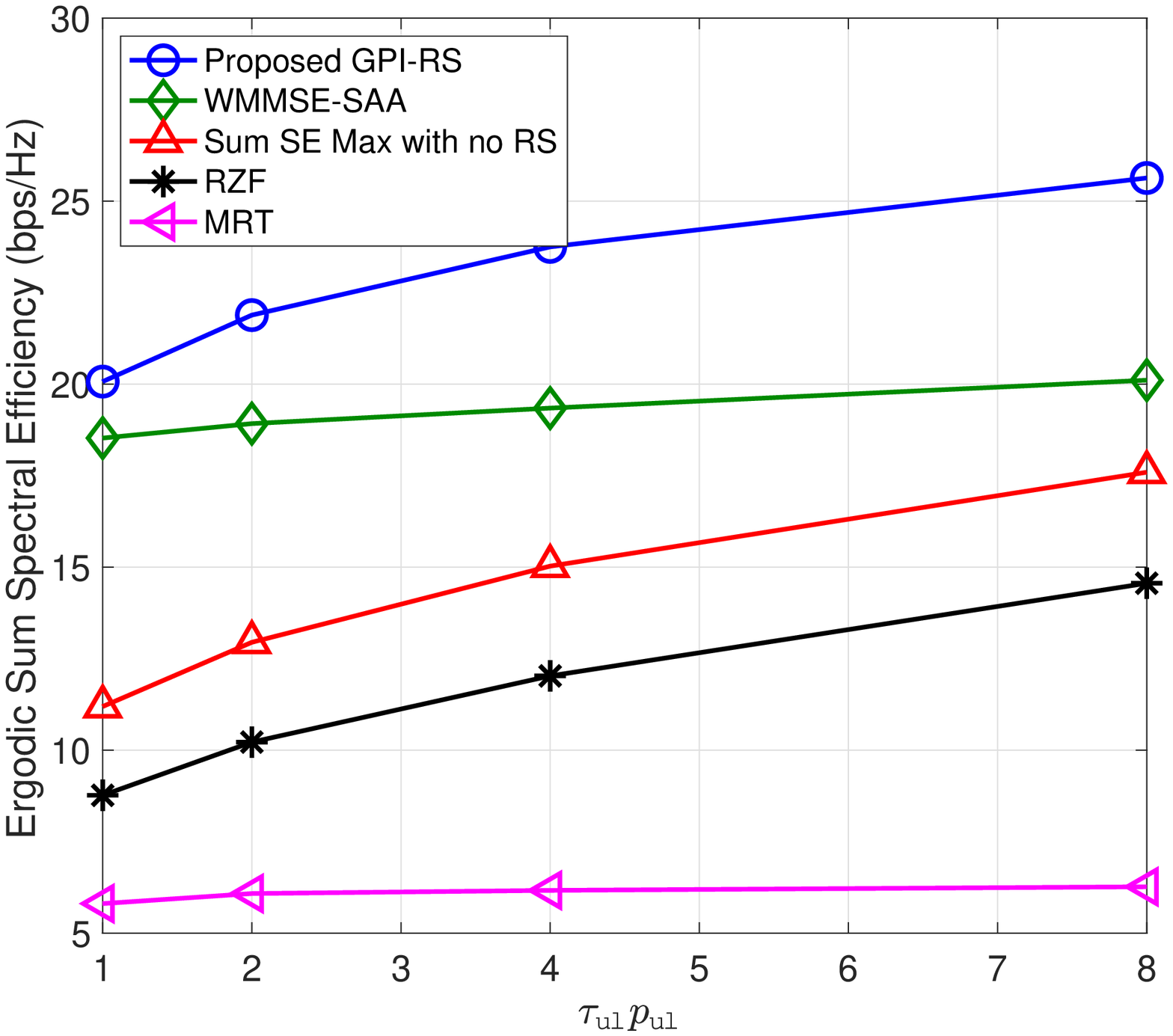}}}     
	\caption{The sum spectral efficiency comparison per CSIT accuracy. The simulation setup is: $N = 6$, $K = 4$, and $\sigma^2 = 1$. The AoA $\theta_k$ is drawn from the uniform distribution. The angular spread is $\Delta_k = \pi/6$ for $k \in \CMcal{K}$.  }
 	\label{fig:perErr}
\end{figure}

{\textbf{Convergence}}: We depict the convergence behavior. The $\alpha$ update process is equal as above: the initial $\alpha$ value is $0.1$ if ${\sf{SNR}} < 15 {\text{dB}}$ and $0.5$ for the rest of the cases. We update $\alpha$ by $0.5$ per every $50$ iterations. We recall that this is for finding the smallest $\alpha$ that guarantees convergence. In Fig.~\ref{fig:converge}, we observe that the GPI-RS converges well with small $\alpha$ in low SNR. On contrary to this, in the high SNR, we need to tune $\alpha$ to make the GPI-RS converge. For instance, when ${\sf{SNR}} = 20 {\text{dB}}$, $\alpha$ needs to be updated one time until convergence, while when ${\sf{SNR}} = 40 {\text{dB}}$, $\alpha$ needs to be updated two times until convergence. Through this observation, we numerically confirm that using the presented $\alpha$ update method, the convergence of the proposed GPI-RS is guaranteed well. 

\begin{figure}[!t]
\centering
$\begin{array}{cc}  
{\resizebox{0.3\columnwidth}{!}
{\includegraphics{ 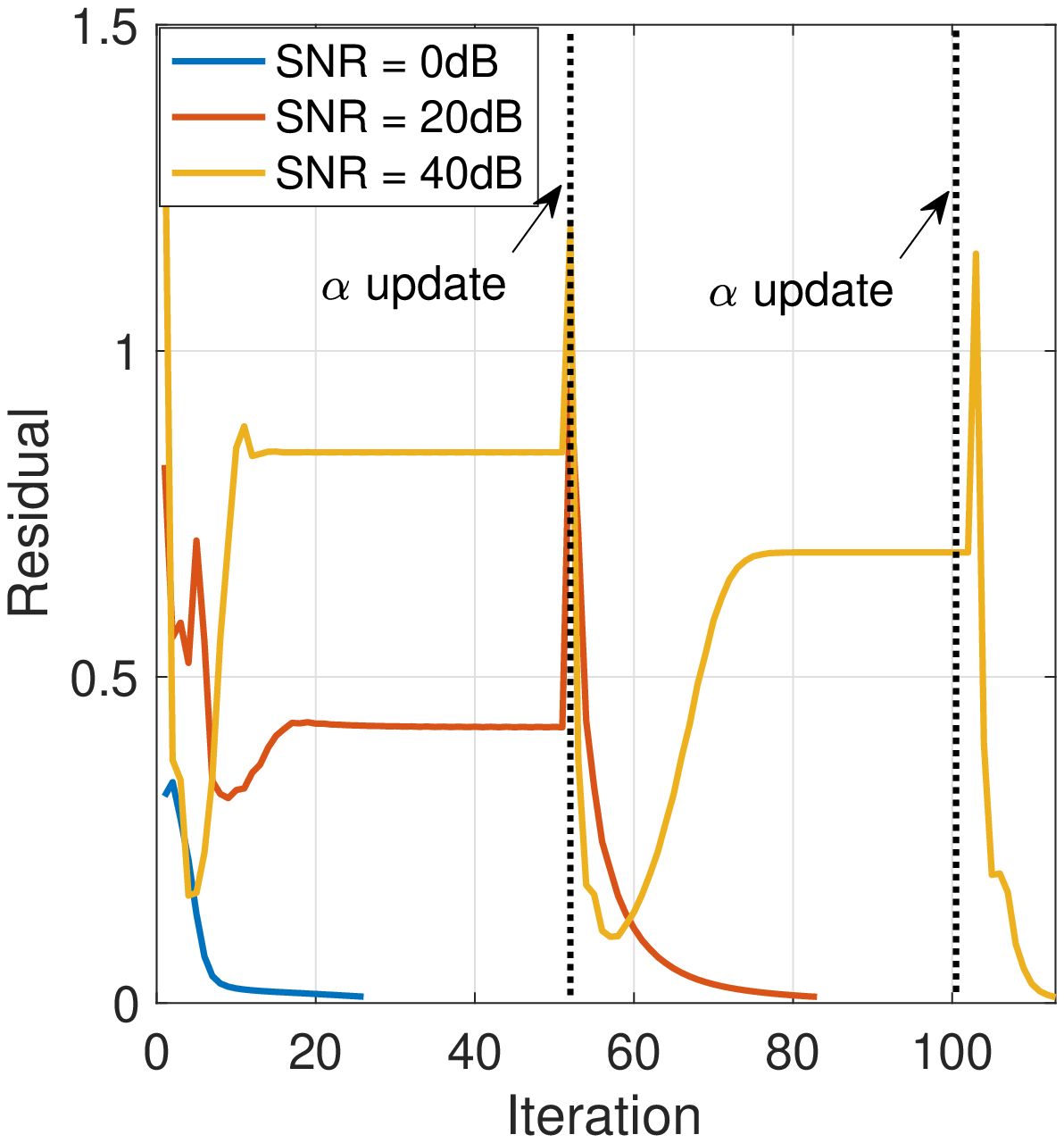}}} &
{\resizebox{0.3\columnwidth}{!}
{\includegraphics{ 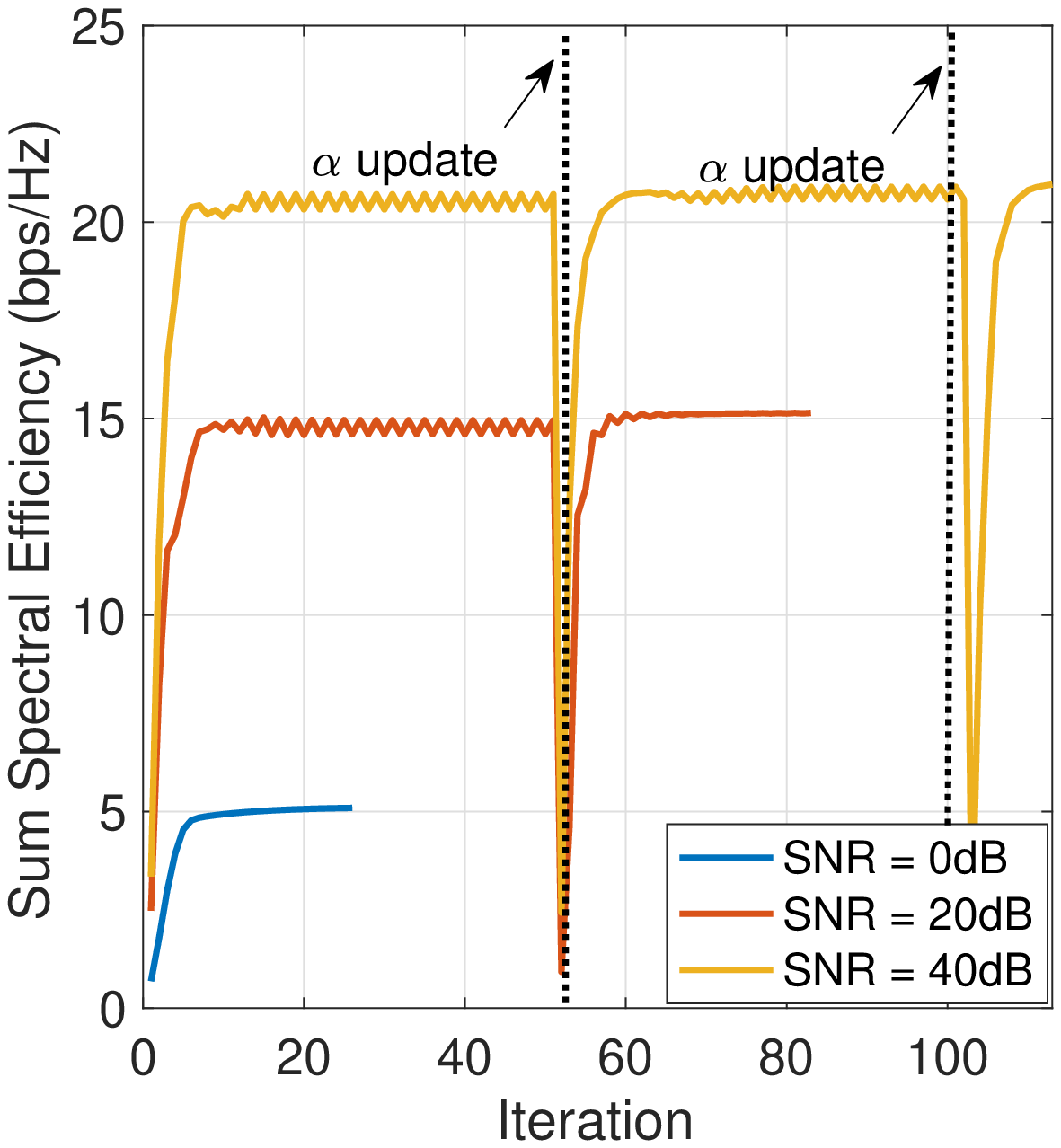}}}  \\  
\mbox{(a)} &
\mbox{(b)}
\end{array}$
\caption{The convergence behavior of the proposed GPI-RS per iteration. In (a) the residual is depicted. In (b). the sum spectral efficiency is drawn. The simulation setup is: $N = 6$, $K = 4$, $\tau_{\sf ul} p_{\sf ul} = 4$ and $\sigma^2 = 1$. The AoA $\theta_k$ is drawn from the uniform distribution and the angular spread is $\Delta_k = \pi/6$ for $k \in \CMcal{K}$. {\color{black}{Residual is defined as $\| \bar {\bf{f}}_{(t)} - \bar {\bf{f}}_{(t-1)} \| $.}}
}
\label{fig:converge}  
\end{figure}

{\color{black}{
\begin{table}[!t]
  \begin{center}
    \caption{Average MATLAB CPU Time (sec)}
    \label{tab:table1}
    \begin{tabular}{c|c|c|c} 
    \hline\hline
    {\text{Setup}} & \textbf{Proposed GPI-RS} & \textbf{WMMSE-SAA} & {\text{Comparison}  ($\%$)} \\
    \hline
      {$6\times 4$} & 2.62 & 35.12 & 7.4$\%$ \\
      {$12\times 8$} & 11.24 & 170.29 & 6.6$\%$ \\
      \hline\hline
    \end{tabular}
  \end{center}
\end{table}

{\textbf{Computation Time}}: 
As a complement to the complexity analysis in Remark~\ref{remark:complexity}, we compare numerical MATLAB computation times between the proposed GPI-RS and the WMMSE-SAA in Table~\ref{tab:table1}. The simulation setups are equivalent to those used to produce Fig.~\ref{fig:compare}. As shown in Table~\ref{tab:table1}, the proposed GPI-RS consumes only $7.4\%$ of the computation time compared to the WMMSE-SAA in the $6 \times 4$ case and $6.6\%$ in the $12 \times 8$ case on average. This dramatic complexity reduction comes from two sources: First, we approximate the non-smooth minimum function via a LogSumExp technique, by which we avoid having $K$ distinct constraints on the common message rate in the optimization problem. Second, by using the GPI-RS, we do not rely on an off-the-shelf optimization toolbox, including CVX.
This result indicates that the proposed method is beneficial in terms of computational complexity, not only in an analytical (big-$\CMcal{O}$) sense but also in a practical (numerical computation time) sense. We note that Table~\ref{tab:table1} is only a rough measure of relative computational complexity. 
}}


 \begin{figure}[!t]
	\centerline{\resizebox{0.45\columnwidth}{!}{\includegraphics{ 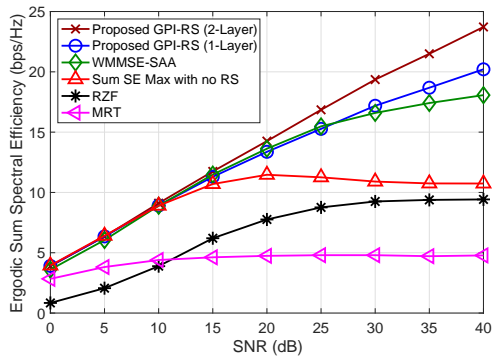}}}
	\caption{The sum spectral efficiency comparison per CSIT accuracy. The simulation setup is: $N = 6$, $K = 4$, $\tau_{\sf ul} p_{\sf ul} = 4$, and $\sigma^2 = 1$. We assume that user $1$, $2$ and user $3$, $4$ are clustered in the same location: $\theta_k = \pi/3$ for $k \in \{1,2\}$ and $\theta_k = 2\pi/3$ for $k \in \{3,4\}$. 
	The angular spread is fixed as $\Delta_k = \pi/6$ for $k \in \CMcal{K}$. }
 	\label{fig:equal}
\end{figure}

{\textbf{Multiple-layer RSMA }}: Next, we evaluate the sum spectral efficiency performance of the multiple-layer RSMA using the generalized GPI-RS. In particular, we assume the $2$-layer RSMA scenario in the $6\times4$ case, wherein one common message, two partial common messages, and private messages are transmitted. Specifically, the partial common messages are $s_{{\sf c}, \CMcal{K}_1}$ and $s_{{\sf c}, \CMcal{K}_2}$, and $\CMcal{K}_1 = \{1,2\}$ and $\CMcal{K}_2 = \{3,4\}$; therefore the partial common message $s_{{\sf c}, \CMcal{K}_1}$ is intended to user $1$ and $2$ and the partial common message $s_{{\sf c}, \CMcal{K}_2}$ is intended to user $3$ and $4$. We construct a favorable channel environment for this message setup, in which users $1$, $2$, and user $3$, $4$ are clustered in the same location, respectively. The performance result is illustrated in Fig.~\ref{fig:equal}, whose caption includes detailed simulation setups. As shown in Fig.~\ref{fig:equal}, the GPI-RS for $2$-layer RSMA achieves the best sum spectral efficiency performance. Specifically, the GPI-RS for $2$-layer RSMA has around $17\%$ gains over the GPI-RS for $1$-layer RSMA and round $31\%$ gains over the WMMSE-SAA. 
This confirms the observation of \cite{dai:twc:16} more in a more rigorous way. We also show that the proposed framework is well extended to multiple-layer RSMA.

\section{Conclusion}
In this paper, we have proposed a novel precoding optimization method for downlink MIMO with RSMA. Aiming to maximize the sum spectral efficiency of the considered system, we have formulated an optimization problem, while a sum spectral efficiency maximization problem regarding linear precoding vectors is infeasible to solve due to its non-convexity and non-smoothness. To resolve this, we have approximated a non-smooth minimum function using the LogSumExp technique and reformulated the problem into a tractable form. We have shown that the first-order optimality condition of the reformulated problem is cast as a NEPv. In order to find the leading eigenvector for the derived condition, we have proposed the GPI-RS. We also have extended the GPI-RS to the multiple-layer RSMA scenario. The simulations have demonstrated that the GPI-RS brings significant spectral efficiency gains in various environments while the associated complexity is small compared to the existing WMMSE-SAA method. 

{\color{black}{
For future work, it is promising to consider the finite blocklength regime where a non-zero decoding probability is induced \cite{choi:iotj:21, poly:tit:10}. In particular, decoding failure on a common message can cause significant interference in the SINR of private messages; therefore, precoders need to be designed carefully to maximize the spectral efficiency. 
In addition, design for physical layer security \cite{lee:arxiv:20, choi:tvt:21} with RSMA is of interest. Considering RSMA in terahertz line-of-sight MIMO environments \cite{do:commmag:21} is also promising. 
}}


\appendices
\section{Proof of Lemma \ref{lem:main}} \label{proof:lem1}
We first derive the KKT condition of the problem \eqref{eq:problem_new}. The corresponding Lagrangian function is defined as
\begin{align}
L(\bar {\bf{f}}) =& \log \left(\frac{1}{K}  \sum_{k = 1}^{K}   \exp \left(\log_2 \left(  \frac{\bar {\bf{f}}^{\sf H} {\bf{A}}_{\sf c} (k) \bar {\bf{f}}}{\bar {\bf{f}}^{\sf H} {\bf{B}}_{\sf c}(k)  \bar {\bf{f}} }  \right)^{-\frac{1}{\alpha}} \right)  \right)^{-\alpha}  + \sum_{k = 1}^{K}  \log_2 \left( \frac{\bar {\bf{f}}^{\sf H} {\bf{A}}_k \bar {\bf{f}}}{\bar {\bf{f}}^{\sf H} {\bf{B}}_k \bar {\bf{f}}}   \right).
\end{align}
To find a stationary point, we take the partial derivatives of $L(\bar {\bf{f}})$ with respect to $\bar {\bf{f}}$ and set it to zero. 
For simplicity, we denote the first and the second part of the Lagrangian function as $L_1(\bar {\bf{f}})$ and $L_2(\bar {\bf{f}})$, respectively. 
Since we have
\begin{align} 
    \partial \left(\frac{\bar {\bf{f}}^{\sf H} {\bf{A}} \bar {\bf{f}}}{\bar {\bf{f}}^{\sf H} {\bf{B}} \bar {\bf{f}}} \right) / \partial \bar {\bf{f}}^{\sf H} = \left( \frac{\bar {\bf{f}}^{\sf H} {\bf{A}} \bar {\bf{f}}}{\bar {\bf{f}}^{\sf H} {\bf{B}} \bar {\bf{f}}} \right) \left[\frac{{\bf{A}} \bar {\bf{f}}}{\bar {\bf{f}}^{\sf H} {\bf{A}} \bar {\bf{f}}} - \frac{ {\bf{B}} \bar {\bf{f}}}{\bar {\bf{f}}^{\sf H} {\bf{B}} \bar {\bf{f}}} \right],
\end{align}
the partial derivative of $L_1(\bar {\bf{f}})$ is obtained using the above calculation:
\begin{align}
&\frac{\partial L_1(\bar {\bf{f}})}{\partial \bar {\bf{f}}^{\sf H}} = 
 \sum_{k = 1}^{K}   \Bigg[ \frac{\exp\left( \! \frac{1}{-\alpha} \! \log_2\left(\frac{\bar {\bf{f}}^{\sf H} {\bf{A}}_{\sf c}(k) \bar {\bf{f}}}{\bar {\bf{f}}^{\sf H} {\bf{B}}_{\sf c}(k) \bar {\bf{f}} } \right) \right)}{\sum_{\ell = 1}^{K} \exp\left(\!\frac{1}{-\alpha}\! \log_2\left(\frac{\bar {\bf{f}}^{\sf H} {\bf{A}}_{\sf c}(\ell) \bar {\bf{f}}}{\bar {\bf{f}}^{\sf H} {\bf{B}}_{\sf c}(\ell) \bar {\bf{f}} } \right)\right)} \!\! \times \! \partial \left(\log_2\left(\frac{\bar {\bf{f}}^{\sf H} {\bf{A}}_{\sf c}(k) \bar {\bf{f}}}{\bar {\bf{f}}^{\sf H} {\bf{B}}_{\sf c}(k) \bar {\bf{f}}}\right) \right)\!/\!\partial {\bf{f}}^{\sf H}   \Bigg]
\nonumber \\
=& \frac{1}{\log 2} \sum_{k = 1}^{K}   \Bigg[ \frac{\exp\left( \frac{1}{-\alpha}  \log_2\left(\frac{\bar {\bf{f}}^{\sf H} {\bf{A}}_{\sf c}(k) \bar {\bf{f}}}{\bar {\bf{f}}^{\sf H} {\bf{B}}_{\sf c}(k) \bar {\bf{f}} } \right) \right)}{\sum_{\ell = 1}^{K} \exp\left(\frac{1}{-\alpha} \log_2\left(\frac{\bar {\bf{f}}^{\sf H} {\bf{A}}_{\sf c}(\ell) \bar {\bf{f}}}{\bar {\bf{f}}^{\sf H} {\bf{B}}_{\sf c}(\ell) \bar {\bf{f}} } \right)\right)}   \times   \left\{\frac{{\bf{A}}_{\sf c}(k) \bar {\bf{f}}}{\bar {\bf{f}}^{\sf H} {\bf{A}}_{\sf c}(k) \bar {\bf{f}}} - \frac{{\bf{B}}_{\sf c}(k) \bar {\bf{f}}}{\bar {\bf{f}}^{\sf H} {\bf{B}}_{\sf c}(k) \bar {\bf{f}}} \right\}  \Bigg]
\end{align}
Similar to this, we calculate $L_2(\bar {\bf{f}})/\partial \bar {\bf{f}}^{\sf H}$ as follows. 
\begin{align} 
\frac{\partial L_2 (\bar {\bf{f}})}{\partial \bar {\bf{f}}^{\sf H}} 
=&\frac{1}{\log 2} \sum_{k = 1}^{K} \left[\frac{{\bf{A}}_k \bar {\bf{f}}}{\bar {\bf{f}}^{\sf H} {\bf{A}}_k \bar {\bf{f}}} - \frac{{\bf{B}}_k \bar {\bf{f}}}{\bar {\bf{f}}^{\sf H} {\bf{B}}_k \bar {\bf{f}}} \right].
\end{align}
The first-order KKT condition holds when 
\begin{align} \label{eq:kkt_inproof}
    &\frac{\partial L_1(\bar {\bf{f}})}{\partial \bar {\bf{f}}^{\sf H}} + \frac{\partial L_2(\bar {\bf{f}})}{\partial \bar {\bf{f}}^{\sf H}} = 0 \\ \Leftrightarrow &
     \sum_{k = 1}^{K}   \Bigg[ \frac{\exp\left( \frac{1}{-\alpha}  \log_2\left(\frac{\bar {\bf{f}}^{\sf H} {\bf{A}}_{\sf c}(k) \bar {\bf{f}}}{\bar {\bf{f}}^{\sf H} {\bf{B}}_{\sf c}(k) \bar {\bf{f}} } \right) \right)}{\sum_{\ell = 1}^{K} \exp\left(\frac{1}{-\alpha} \log_2\left(\frac{\bar {\bf{f}}^{\sf H} {\bf{A}}_{\sf c}(\ell) \bar {\bf{f}}}{\bar {\bf{f}}^{\sf H} {\bf{B}}_{\sf c}(\ell) \bar {\bf{f}} } \right)\right)}   \times  \left\{\frac{{\bf{A}}_{\sf c}(k) \bar {\bf{f}}}{\bar {\bf{f}}^{\sf H} {\bf{A}}_{\sf c}(k) \bar {\bf{f}}} - \frac{{\bf{B}}_{\sf c}(k) \bar {\bf{f}}}{\bar {\bf{f}}^{\sf H} {\bf{B}}_{\sf c}(k) \bar {\bf{f}}} \right\}  \Bigg] +  \sum_{k = 1}^{K} \left[\frac{{\bf{A}}_k \bar {\bf{f}}}{\bar {\bf{f}}^{\sf H} {\bf{A}}_k \bar {\bf{f}}} - \frac{{\bf{B}}_k \bar {\bf{f}}}{\bar {\bf{f}}^{\sf H} {\bf{B}}_k \bar {\bf{f}}} \right] = 0.
\end{align}
Defining ${\bf{A}}_{\sf KKT}(\bar {\bf{f}})$, ${\bf{B}}_{\sf KKT}(\bar {\bf{f}})$, and $\lambda(\bar {\bf{f}})$ as 
\begin{align} \label{eq:lem_A_kkt_inproof}
    &{\bf{A}}_{\sf KKT}(\bar {\bf{f}}) =  \lambda_{\sf num} (\bar {\bf{f}}) \times  \sum_{k = 1}^{K}  \left[ \frac{\exp\left( \frac{1}{-\alpha}  \log_2\left(\frac{\bar {\bf{f}}^{\sf H} {\bf{A}}_{\sf c}(k) \bar {\bf{f}}}{\bar {\bf{f}}^{\sf H} {\bf{B}}_{\sf c}(k) \bar {\bf{f}} } \right) \right)}{\sum_{\ell = 1}^{K} \exp\left(\frac{1}{-\alpha} \log_2\left(\frac{\bar {\bf{f}}^{\sf H} {\bf{A}}_{\sf c}(\ell) \bar {\bf{f}}}{\bar {\bf{f}}^{\sf H} {\bf{B}}_{\sf c}(\ell) \bar {\bf{f}} } \right)\right)} \frac{{\bf{A}}_{\sf c}(k)}{\bar {\bf{f}}^{\sf H} {\bf{A}}_{\sf c}(k) \bar {\bf{f}}} + \frac{{\bf{A}}_k}{\bar {\bf{f}}^{\sf H} {\bf{A}}_k \bar {\bf{f}}} \right] , 
    \end{align}
    \begin{align}
    &{\bf{B}}_{\sf KKT}(\bar {\bf{f}}) = \lambda_{\sf den} (\bar {\bf{f}}) \times  \sum_{k = 1}^{K}  \left[ \frac{\exp\left( \frac{1}{-\alpha}  \log_2\left(\frac{\bar {\bf{f}}^{\sf H} {\bf{A}}_{\sf c}(k) \bar {\bf{f}}}{\bar {\bf{f}}^{\sf H} {\bf{B}}_{\sf c}(k) \bar {\bf{f}} } \right) \right)}{\sum_{\ell = 1}^{K} \exp\left(\frac{1}{-\alpha} \log_2\left(\frac{\bar {\bf{f}}^{\sf H} {\bf{A}}_{\sf c}(\ell) \bar {\bf{f}}}{\bar {\bf{f}}^{\sf H} {\bf{B}}_{\sf c}(\ell) \bar {\bf{f}} } \right)\right)} \frac{{\bf{B}}_{\sf c}(k)}{\bar {\bf{f}}^{\sf H} {\bf{B}}_{\sf c}(k) \bar {\bf{f}}} + \frac{{\bf{B}}_k}{\bar {\bf{f}}^{\sf H} {\bf{B}}_k \bar {\bf{f}}} \right] , \label{eq:lem_B_kkt_inproof} \\
    &\lambda(\bar {\bf{f}}) = \left\{\frac{1}{K  } \sum_{k = 1}^{K} \exp\left( \log_2
    \left(\frac{\bar {\bf{f}}^{\sf H} {\bf{A}}_{\sf c}(k) \bar {\bf{f}}}{\bar {\bf{f}}^{\sf H} {\bf{B}}_{\sf c}(k) {\bf{f}}} \right)^{-\frac{1}{\alpha }} \right)\right\}^{-\frac{\alpha}{\log_2 e}}   \times \prod_{k = 1}^{K} \left(\frac{\bar {\bf{f}}^{\sf H} {\bf{A}}_k \bar {\bf{f}}}{\bar {\bf{f}}^{\sf H} {\bf{B}}_k \bar {\bf{f}}} \right) = \frac{\lambda_{\sf num} (\bar {\bf{f}})}{\lambda_{\sf den} (\bar {\bf{f}})}, \label{eq:lem_lambda_inproof}
\end{align}
the first-order KKT condition is rearranged as
\begin{align}
 & {\bf{A}}_{\sf KKT}(\bar {\bf{f}}) \bar {\bf{f}} = \lambda(\bar {\bf{f}}) {\bf{B}}_{\sf KKT} (\bar {\bf{f}}) \bar {\bf{f}} \Leftrightarrow 
 {\bf{B}}_{\sf KKT} (\bar {\bf{f}})^{-1}{\bf{A}}_{\sf KKT}(\bar {\bf{f}}) \bar {\bf{f}} = \lambda(\bar {\bf{f}}) \bar {\bf{f}}.
\end{align}
This completes the proof.
\qed

\bibliographystyle{IEEEtran}
\bibliography{ref_rsMIMO}

\begin{thebibliography}{10}
\providecommand{\url}[1]{#1}
\csname url@samestyle\endcsname
\providecommand{\newblock}{\relax}
\providecommand{\bibinfo}[2]{#2}
\providecommand{\BIBentrySTDinterwordspacing}{\spaceskip=0pt\relax}
\providecommand{\BIBentryALTinterwordstretchfactor}{4}
\providecommand{\BIBentryALTinterwordspacing}{\spaceskip=\fontdimen2\font plus
\BIBentryALTinterwordstretchfactor\fontdimen3\font minus
  \fontdimen4\font\relax}
\providecommand{\BIBforeignlanguage}[2]{{%
\expandafter\ifx\csname l@#1\endcsname\relax
\typeout{** WARNING: IEEEtran.bst: No hyphenation pattern has been}%
\typeout{** loaded for the language `#1'. Using the pattern for}%
\typeout{** the default language instead.}%
\else
\language=\csname l@#1\endcsname
\fi
#2}}
\providecommand{\BIBdecl}{\relax}
\BIBdecl

\bibitem{park:wcnc:21}
J.~Park, J.~Choi, N.~Lee, W.~Shin, and H.~V. Poor, ``Sum spectral efficiency
  optimization for rate splitting in downlink {MU-MISO}: {A} generalized power
  iteration approach,'' in \emph{{Proc. IEEE Wireless Commun. and Netw. Conf.
  Workshop}}, 2021, pp. 1--6.

\bibitem{caire:tit:03}
G.~{Caire} and S.~{Shamai}, ``On the achievable throughput of a multiantenna
  {Gaussian} broadcast channel,'' \emph{IEEE Trans. Inf. Theory}, vol.~49,
  no.~7, pp. 1691--1706, Jul. 2003.

\bibitem{chris:twc:08}
S.~S. {Christensen}, R.~{Agarwal}, E.~D. {Carvalho}, and J.~M. {Cioffi},
  ``Weighted sum-rate maximization using weighted {MMSE} for {MIMO-BC}
  beamforming design,'' \emph{IEEE Trans. Wireless Commun.}, vol.~7, no.~12,
  pp. 4792--4799, Dec. 2008.

\bibitem{choi:twc:20}
J.~{Choi}, N.~{Lee}, S.~{Hong}, and G.~{Caire}, ``Joint user selection, power
  allocation, and precoding design with imperfect {CSIT} for multi-cell
  {MU-MIMO} downlink systems,'' \emph{IEEE Trans. Wireless Commun.}, vol.~19,
  no.~1, pp. 162--176, 2020.

\bibitem{jindal:tit:06}
N.~{Jindal}, ``{MIMO} broadcast channels with finite-rate feedback,''
  \emph{IEEE Trans. Inf. Theory}, vol.~52, no.~11, pp. 5045--5060, 2006.

\bibitem{park:twc:16}
J.~{Park}, N.~{Lee}, J.~G. {Andrews}, and R.~W. {Heath}, ``On the optimal
  feedback rate in interference-limited multi-antenna cellular systems,''
  \emph{IEEE Trans. Wireless Commun.}, vol.~15, no.~8, pp. 5748--5762, 2016.

\bibitem{joudeh:16:tcom}
H.~{Joudeh} and B.~{Clerckx}, ``Sum-rate maximization for linearly precoded
  downlink multiuser {MISO} systems with partial {CSIT}: {A} rate-splitting
  approach,'' \emph{IEEE Trans. Commun.}, vol.~64, no.~11, pp. 4847--4861,
  2016.

\bibitem{joudeh:16:tsp}
------, ``Robust transmission in downlink multiuser {MISO} systems: {A}
  rate-splitting approach,'' \emph{IEEE Trans. Signal Process.}, vol.~64,
  no.~23, pp. 6227--6242, 2016.

\bibitem{dai:16:twc}
M.~{Dai}, B.~{Clerckx}, D.~{Gesbert}, and G.~{Caire}, ``A rate splitting
  strategy for massive {MIMO} with imperfect {CSIT},'' \emph{IEEE Trans.
  Wireless Commun.}, vol.~15, no.~7, pp. 4611--4624, 2016.

\bibitem{li:jsac:20}
Z.~{Li}, C.~{Ye}, Y.~{Cui}, S.~{Yang}, and S.~{Shamai}, ``Rate splitting for
  multi-antenna downlink: {Precoder} design and practical implementation,''
  \emph{IEEE J. Sel. Areas Commun.}, vol.~38, no.~8, pp. 1910--1924, 2020.

\bibitem{mao:tcom:20}
Y.~{Mao} and B.~{Clerckx}, ``Beyond dirty paper coding for multi-antenna
  broadcast channel with partial {CSIT}: {A} rate-splitting approach,''
  \emph{IEEE Trans. Commun.}, vol.~68, no.~11, pp. 6775--6791, 2020.

\bibitem{han:tit:81}
{Te Han} and K.~{Kobayashi}, ``A new achievable rate region for the
  interference channel,'' \emph{IEEE Trans. Inf. Theory}, vol.~27, no.~1, pp.
  49--60, 1981.

\bibitem{etkin:tit:08}
R.~H. {Etkin}, D.~N.~C. {Tse}, and H.~{Wang}, ``Gaussian interference channel
  capacity to within one bit,'' \emph{IEEE Trans. Inf. Theory}, vol.~54,
  no.~12, pp. 5534--5562, 2008.

\bibitem{clercks:commmag:16}
B.~{Clerckx}, H.~{Joudeh}, C.~{Hao}, M.~{Dai}, and B.~{Rassouli}, ``Rate
  splitting for {MIMO} wireless networks: {A} promising {PHY}-layer strategy
  for {LTE} evolution,'' \emph{IEEE Commun. Mag.}, vol.~54, no.~5, pp. 98--105,
  2016.

\bibitem{yang:13:tit}
S.~{Yang}, M.~{Kobayashi}, D.~{Gesbert}, and X.~{Yi}, ``Degrees of freedom of
  time correlated {MISO} broadcast channel with delayed {CSIT},'' \emph{IEEE
  Trans. Inf. Theory}, vol.~59, no.~1, pp. 315--328, 2013.

\bibitem{hao:tcom:15}
C.~{Hao}, Y.~{Wu}, and B.~{Clerckx}, ``Rate analysis of two-receiver {MISO}
  broadcast channel with finite rate feedback: A rate-splitting approach,''
  \emph{IEEE Trans. Commun.}, vol.~63, no.~9, pp. 3232--3246, 2015.

\bibitem{yalcin:twc:21}
A.~Z. {Yalçın} and Y.~{Yapıcı}, ``Max-min fair beamforming for cooperative
  multigroup multicasting with rate-splitting,'' \emph{IEEE Trans. Wireless
  Commun.}, vol.~20, no.~1, pp. 254--268, 2021.

\bibitem{dizdar:pimrc:20}
O.~{Dizdar}, Y.~{Mao}, W.~{Han}, and B.~{Clerckx}, ``Rate-splitting multiple
  access for downlink multi-antenna communications: {Physical} layer design and
  link-level simulations,'' in \emph{Proc. IEEE Int. Symp. Pers., Indoor Mobile
  Radio Commun.}, 2020, pp. 1--6.

\bibitem{yang:icc:20}
Z.~{Yang}, M.~{Chen}, W.~{Saad}, and M.~{Shikh-Bahaei}, ``Downlink sum-rate
  maximization for rate splitting multiple access {(RSMA)},'' in \emph{Proc.
  IEEE Int. Conf. Comm.}, 2020, pp. 1--6.

\bibitem{dai:twc:16}
M.~{Dai}, B.~{Clerckx}, D.~{Gesbert}, and G.~{Caire}, ``A rate splitting
  strategy for massive {MIMO} with imperfect {CSIT},'' \emph{IEEE Trans.
  Wireless Commun.}, vol.~15, no.~7, pp. 4611--4624, Jul. 2016.

\bibitem{adhi:tit:13}
A.~{Adhikary}, J.~{Nam}, J.~{Ahn}, and G.~{Caire}, ``Joint spatial division and
  multiplexing - {T}he large-scale array regime,'' \emph{IEEE Trans. Inf.
  Theory}, vol.~59, no.~10, pp. 6441--6463, Oct. 2013.

\bibitem{mao:tcom:19}
Y.~{Mao}, B.~{Clerckx}, and V.~O.~K. {Li}, ``Rate-splitting for multi-antenna
  non-orthogonal unicast and multicast transmission: {Spectral} and energy
  efficiency analysis,'' \emph{IEEE Trans. Commun.}, vol.~67, no.~12, pp.
  8754--8770, 2019.

\bibitem{papa:tvt:17}
A.~{Papazafeiropoulos}, B.~{Clerckx}, and T.~{Ratnarajah}, ``Rate-splitting to
  mitigate residual transceiver hardware impairments in massive {MIMO}
  systems,'' \emph{IEEE Trans. Veh. Technol.}, vol.~66, no.~9, pp. 8196--8211,
  2017.

\bibitem{xu:jsac:21}
C.~Xu, B.~Clerckx, S.~Chen, Y.~Mao, and J.~Zhang, ``Rate-splitting multiple
  access for multi-antenna joint radar and communications,'' \emph{IEEE J. Sel.
  Areas Commun.}, pp. 1--1, 2021.

\bibitem{zeng:tvt:19}
J.~{Zeng}, T.~{Lv}, W.~{Ni}, R.~P. {Liu}, N.~C. {Beaulieu}, and Y.~J. {Guo},
  ``Ensuring max–min fairness of {UL} {SIMO-NOMA}: {A} rate splitting
  approach,'' \emph{IEEE Trans. Veh. Technol.}, vol.~68, no.~11, pp.
  11\,080--11\,093, 2019.

\bibitem{hao:tcom:17}
C.~{Hao} and B.~{Clerckx}, ``{MISO} networks with imperfect {CSIT}: {A}
  topological rate-splitting approach,'' \emph{IEEE Trans. Commun.}, vol.~65,
  no.~5, pp. 2164--2179, 2017.

\bibitem{krivochiza:access:19}
J.~Krivochiza, J.~Merlano~Duncan, S.~Andrenacci, S.~Chatzinotas, and
  B.~Ottersten, ``{FPGA} acceleration for computationally efficient
  symbol-level precoding in multi-user multi-antenna communication systems,''
  \emph{IEEE Access}, vol.~7, pp. 15\,509--15\,520, 2019.

\bibitem{cai:siam:18}
Y.~Cai, L.-H. Zhang, Z.~Bai, and R.-C. Li, ``On an eigenvector-dependent
  nonlinear eigenvalue problem,'' \emph{{SIAM J. Matrix Anal. Appl.}}, vol.~39,
  no.~3, pp. 1360--1382, 2018.

\bibitem{yin:jsac:13}
H.~{Yin}, D.~{Gesbert}, M.~{Filippou}, and Y.~{Liu}, ``A coordinated approach
  to channel estimation in large-scale multiple-antenna systems,'' \emph{IEEE
  J. Sel. Areas Commun.}, vol.~31, no.~2, pp. 264--273, 2013.

\bibitem{patil:tcom:18}
P.~{Patil}, B.~{Dai}, and W.~{Yu}, ``Hybrid data-sharing and compression
  strategy for downlink cloud radio access network,'' \emph{IEEE Trans.
  Commun.}, vol.~66, no.~11, pp. 5370--5384, Nov. 2018.

\bibitem{shen2010dual}
C.~Shen and H.~Li, ``{On the dual formulation of boosting algorithms},''
  \emph{IEEE Trans. Pattern Anal. Mach. Intell.}, vol.~32, no.~12, pp.
  2216--2231, 2010.

\bibitem{joude:twc:17}
H.~{Joudeh} and B.~{Clerckx}, ``Rate-splitting for max-min fair multigroup
  multicast beamforming in overloaded systems,'' \emph{IEEE Trans. Wireless
  Commun.}, vol.~16, no.~11, pp. 7276--7289, Nov. 2017.

\bibitem{peng:tsp:17}
G.~Peng, L.~Liu, P.~Zhang, S.~Yin, and S.~Wei, ``Low-computing-load,
  high-parallelism detection method based on {Chebyshev} iteration for massive
  {MIMO} systems with {VLSI} architecture,'' \emph{IEEE Trans. Signal
  Process.}, vol.~65, no.~14, pp. 3775--3788, 2017.

\bibitem{zhu:icc:15}
D.~Zhu, B.~Li, and P.~Liang, ``On the matrix inversion approximation based on
  {Neumann} series in massive {MIMO} systems,'' in \emph{IEEE International
  Conf. on Commun. (ICC)}, 2015, pp. 1763--1769.

\bibitem{choi:iotj:21}
J.~Choi and J.~Park, ``{MIMO} design for {Internet-of-Things}: {Joint}
  optimization of spectral efficiency and error probability in finite
  blocklength regime,'' \emph{IEEE Internet of Things J.}, pp. 1--1, 2021.

\bibitem{poly:tit:10}
Y.~{Polyanskiy}, H.~V. {Poor}, and S.~{Verdu}, ``Channel coding rate in the
  finite blocklength regime,'' \emph{IEEE Trans. Inf. Theory}, vol.~56, no.~5,
  pp. 2307--2359, May 2010.

\bibitem{lee:arxiv:20}
\BIBentryALTinterwordspacing
K.~Lee, J.~Choi, D.~K. Kim, and J.~Park, ``Secure transmission for hierarchical
  information accessibility in downlink {MU-MIMO},'' \emph{ArXiv}, 2021.
  [Online]. Available: \url{https://arxiv.org/abs/2109.07727}
\BIBentrySTDinterwordspacing

\bibitem{choi:tvt:21}
J.~Choi and J.~Park, ``Sum secrecy spectral efficiency maximization in downlink
  {MU-MIMO}: {Colluding} eavesdroppers,'' \emph{IEEE Trans. Veh. Technol.},
  vol.~70, no.~1, pp. 1051--1056, 2021.

\bibitem{do:commmag:21}
H.~Do, S.~Cho, J.~Park, H.-J. Song, N.~Lee, and A.~Lozano, ``Terahertz
  line-of-sight {MIMO} communication: {Theory} and practical challenges,''
  \emph{IEEE Commun. Mag.}, vol.~59, no.~3, pp. 104--109, 2021.

\end{thebibliography}

\end{document}